\documentclass[11pt]{article}
\def\withcolors{1}
\def\withnotes{0}

\usepackage{ccanonne}
\usepackage{palatino}
\title{Uniformity Testing under User-Level Local Privacy}
\author{Cl\'ement L. Canonne\thanks{The University of Sydney. Email: \email{clement.canonne@sydney.edu.au}.} 
\and Abigail Gentle\thanks{The University of Sydney. Email: \email{abigail.gentle@sydney.edu.au}.} 
\and Vikrant Singhal\thanks{OpenDP, Harvard University. Email: \email{vikrant@seas.harvard.edu}.}}

\begin{document}

\allowdisplaybreaks

\maketitle
\begin{abstract}
   We initiate the study of distribution testing under \emph{user-level} local differential privacy, where each of $n$ users contributes $m$ samples from the unknown underlying distribution. This setting, albeit very natural, is significantly more challenging that the usual locally private setting, as for the same parameter $\varepsilon$ the privacy guarantee must now apply to a full batch of $m$ data points. While some recent work consider distribution \emph{learning} in this user-level setting, nothing was known for even the most fundamental testing task, uniformity testing (and its generalization, identity testing).

We address this gap, by providing (nearly) sample-optimal user-level LDP algorithms for uniformity and identity testing. Motivated by practical considerations, our main focus is on the private-coin, symmetric setting, which does not require users to share a common random seed nor to have been assigned a globally unique identifier.
\end{abstract}

\section{Introduction}
\label{sec:introduction}

We consider the problem of uniformity testing (equivalently, identity testing~\cite{Goldreich:16, CanonneTopicsDT2022}) of distributions in the setting where each of $\ns$ distributed users hold $\ms$ independent and identically distributed observations from some unknown common distribution. This naturally captures many real-world statistical scenarios, such as when data is distributed among many users' personal devices.

Consider the testing equivalent of the problem specified in~\cite{apple_differential_privacy_team_learning_2017}, where the goal is to learn users' most-used emoji. In this practical deployment, users are queried once per-day with a prespecified privacy budget. Of course, people typically use a lot more than a single emoji during that time period, and so one would hope to obtain information about many emojis at once, from each user. Yet, despite users ``sampling'' from the ``distribution'' of emoji's multiple times per day, the $\ms=1$ setting explored in earlier literature can only make use of \emph{one} sample.

\begin{framed}\itshape
    \noindent Can we leverage the fact that each user holds \emph{many} samples to test the underlying distribution more efficiently, while still preserving the privacy of each individual \emph{as a whole}?
\end{framed}

To formalize this question, we work in the framework of Differential Privacy (DP)~\cite{DMNS:06}, specifically Local Differential Privacy (LDP)~\cite{Kasiviswanathan11,Warner:65}, where data is made private before it leaves the device. This setting is of practical interest as data collection has grown massively, and many users are hesitant to trust a central curator with collecting and storing their data non-privately. Furthermore it has received theoretical attention as a well-parameterized model of learning under constrained information per-sample~\cite{AcharyaCFST21}. However, the usual setting of ``item-level'' (\ie single-sample) LDP is ill-suited to our goal, which is to capture the fact that each user can contribute many samples: na\"ively, this would correspond to viewing the data of each user as an $\ms$-tuple, blowing up the domain size from $\ab$ to $\ab^{\ms}$ and leading to severely suboptimal algorithms. Instead, we will work in the more stringent setting of \emph{user-level} LDP, whereby the privacy guarantee applies to the whole data held by any single user (see~\cref{fig:placeholder} and~\cref{sec:prelims} for an illustration and definition).

\begin{figure}[htbp]
    \centering
    \includegraphics[width=0.5\linewidth]{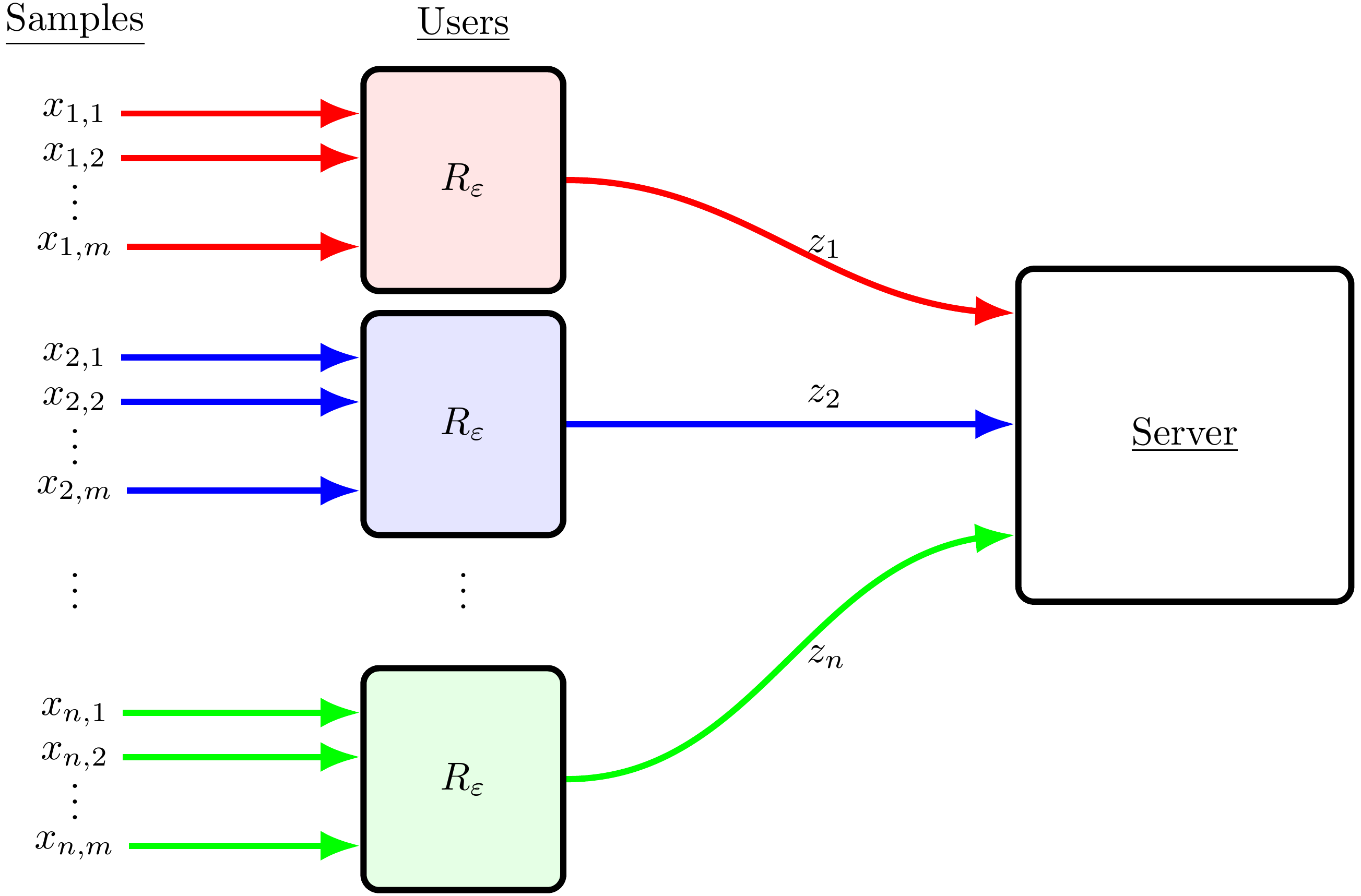}
    \caption{Graphical representation of \emph{user-level} local differential privacy. Each user holds $\ms$ samples of some unknown distribution $\p$, and must be guaranteed privacy of all $\ms$ samples at once. As pictured, each user sends a single private message to the server, which must somehow carry information about \emph{all} of their samples.}
    \label{fig:placeholder}
\end{figure}

As we elaborate in~\cref{ssec:overview}, user-level LDP is a much more stringent and challenging setting, as (roughly speaking) the algorithm's output must be remain similar even when the data of an entire user~--~a full batch of $\ms$ samples!~--~is changed arbitrarily.

As our focus is on capturing practically motivated settings, we also pay special attention to the \emph{type} of algorithm we design, and how realistic their deployment would be. For instance, \emph{adaptive} protocols, in which users interact sequentially and adapt their output to the previous message observed, are typically undesirable, as they introduce latency and a host of technical implementation challenges.

For this reason, we consider \emph{non-adaptive} protocols for this problem; which themselves come in several varieties. In particular a protocol can be \emph{public-coin} or \emph{private-coin}, and \emph{symmetric} or \emph{asymmetric}. A public-coin protocol assumes the existence of a shared random seed between all participants, while a private-coin protocol does not. Symmetric protocols are such that each user runs the same algorithm locally with the same parameters, while an asymmetric protocol allows some variation between users, as decided by the central data collector or curator. These models capture the level of coordination between the curator and the distributed users such that a \emph{public-coin asymmetric} protocol is the most coordinated, and a \emph{private-coin symmetric} protocol is the least. For this reason we will progress through the models from most-coordination to least, as the latter model is more practical and thus preferable. Of course, if the $\ms=1$ case is any indication, this practicality may come at a cost, in terms of per-user communication~\cite{AcharyaS:19} required or overall sample complexity~\cite{AcharyaCFST21}.

\begin{framed}\itshape
    \noindent Can we obtain \emph{symmetric}, \emph{private-coin} testing algorithms for user-level LDP uniformity testing? What is the cost of achieving these two desirable features, in terms of sample and communication~complexity?
\end{framed}

To ground these distinctions and motivate the question further, one can consider the following tasks that are easy in one model, but hard in its complement. A public-coin protocol may assume that all users can apply the same randomly sampled hash function to their data, while a private-coin protocol requires that the hash is either determined in advance, or sampled independently and sent by each user to the server. An asymmetric protocol may easily partition the users into groups of equal size, while to achieve the same effect a symmetric protocol must have users randomly partition themselves and send which group they landed in. This runs into the canonical coupon collector problem, where some groups will be under-sampled, requiring a more thorough analysis, or more samples.

\subsection{Overview of our Results and Techniques}
\label{ssec:overview}

In what follows, we focus on \emph{uniformity} testing, that is, the task of testing whether the unknown distribution $\p$ is uniform on the domain, or at distance at least $\dst$ (in total variation distance) from uniform. As discussed in~\cref{ssec:related}, by a now-standard reduction this directly implies the analogous statements for \emph{identity} testing, where the reference distribution is a fixed, known reference $\q$ instead of the uniform distribution $\uniform_{\ab}$.

Our setting (formally described in~\cref{sec:prelims}) involves $\ns$ users, each holding $\ms\geq 1$ \iid samples from an unknown distribution $\p$ over a known domain of size $\ab \gg 1$. They engage in a distributed protocol, either public- or private-coin, to enable a central server to perform uniformity testing on this underlying distribution in a locally private manner, with privacy parameter $\priv > 0$. Our specific focus will be on the regime $1\leq \ms \leq \sqrt{\ab}/\dst^2$, arguably the most natural and relevant, where each user has a possibly large number of observations, but not enough that they could run a uniformity testing algorithm by themselves:\footnote{Indeed, if $\ms \gg \sqrt{\ab}/\dst^2$, then each user has enough samples to perform uniformity testing by themselves, and so the question boils down to communicating the answer to the server in a locally private manner, for which $\ns = O(1/\priv^2)$ users is sufficient. When $\ms$ is even larger (specifically, at least the sample complexity of uniformity testing in the central DP model, then $\ns=1$ is enough: a single user can run the DP algorithm and send the outcome.} we implicitly place ourselves in this regime throughout.

\bigskip

Our first ``warm-up'' result is a public-coin algorithm\footnote{Throughout, and in line with the local privacy literature, we will use ``algorithm'' and ``protocol'' interchangeably.} which shows that, for uniformity testing, user-level LDP with $\ms$ samples for each of $\ns$ user behaves similarly to LDP  with \emph{one} sample for each of $\ms\ns$ users:
\begin{theorem}[{See~\cref{thm:pubcoin:ldp:1bit}}]
    \label{thm:pubcoin:ldp:1bit:intro}
There exists an \emph{asymmetric, public-coin, user-level} locally differentially private algorithm for uniformity testing (over domain of size $\ab$) which on privacy and distance parameters $\priv>0$ and $\dst\in(0,1]$ takes
\begin{equation*}
    \ns = \bigO{\frac{\ab}{\ms\dst^2\priv^2}}
\end{equation*}
users, each holding $\ms$ \iid samples from the unknown distribution and sending \emph{one} bit of communication. %
\end{theorem}
Whether this result seems unexpected or not to the reader, one interesting aspect is that this result essentially follows from combining, in a very simple way, two known algorithmic building blocks: the first is \emph{domain compression}~\cite{ACT:19,ACHST:20}, a type of hashing which, using shared randomness, allows us to to reduce the domain size without increasing the total variation distance ``too much''. The second is the user-level LDP coin \emph{learning} (asymmetric) protocol of~\cite{AcharyaLS23}, which lets us learn the bias of a single Bernoulli (and so, \emph{a fortiori}, test it). Combining the two is then rather straightforward, modulo some bookkeeping details: (1) use domain compression to reduce the uniformity testing instance from $(\ab,\dst)$ to $(2,\dst/\sqrt{\ab})$; then (2) (privately, with $\ms$ samples) learn the bias of the resulting ``coin'' to accuracy $\dst/(2\sqrt{\ab})$.

In view of the simplicity of the first algorithm, it is natural to expect the private-coin algorithm to be similarly straightforward. As we discuss shortly, this expectation turns out to be significantly optimistic: nonetheless, we are able to establish the following, providing an analogous result in the private-coin setting. (To interpret it, it is useful to remember that, for $\ms=1$, the sample complexity of private-coin locally private uniformity testing is known to be $\Theta(\ab^{3/2}/(\dst^2\priv^2))$.)
\begin{theorem}[{Informal version of~\cref{thm:private:symmetric:combined:testing}}]
    \label{thm:prcoin:ldp:1bit:informal}
There exists a \emph{symmetric, private-coin, user-level} locally differentially private algorithm for uniformity testing (over domain of size $\ab$) which on privacy and distance parameters $\priv>0$ and $\dst\in(0,1]$ takes
\begin{equation*}
    \ns = \tildeO{\frac{\ab^{3/2}}{\ms\dst^2\priv^2}}
\end{equation*}
users, each holding $\ms$ \iid samples from the unknown distribution and sending $O(\log\ab)$ bits of communication. %
\end{theorem}

\noindent Furthermore, by making the algorithm of~\cref{thm:prcoin:ldp:1bit} asymmetric, one can reduce the communication complexity to a \emph{single} bit. One may wonder whether this trade-off between (a)symmetry and communication complexity is inherent: by a simple modification of an argument of Acharya and Sun~\cite{AcharyaS:19} (originally in the context of (item-level) locally private distribution \emph{learning}), we provide strong evidence that this trade-off is necessary for at least some parameter regime,\footnote{We cannot hope to show this trade-off for \emph{all} values of $\ms$, since for $\ms \gg \sqrt{\ab}/\dst^2$, as discussed before, there is a trivial, private-coin, symmetric protocol with a single bit per user.} showing that this is the case for $\ms=1$:
\begin{proposition}
    \label{thm:lower:bound:communication}
Any \emph{symmetric, private-coin} algorithm for uniformity testing (over domain of size $\ab$) with distance parameter $\dst\leq 1/\ab$ and $\ms=1$ sample per user requires at least $\log_2\ab$ bits of communication per user. (This holds regardless of whether the algorithm is locally private or not.)
\end{proposition}
\noindent For completeness, we give a proof of this result in~\cref{app:lowerbounds}. 
While only stated for very small $\dst$, this result shows that one cannot achieve constant communication complexity per user with symmetric private-coin protocols. %

Underlying our results is a technical lemma, likely of independent interest, which gives a way to compress many samples into one bit. This compression, which is likely to also find applications in the (non-private) bandwidth-constrained model, enables us to reduce the task to the much simpler task of privatizing a single bit, rather than a higher-dimensional message. Phrased in terms of differential privacy, this considerably reduces the sensitivity of our randomizers.

\begin{lemma}[Informal statement of~\cref{lemma:epluribusunum}]
    Let $X\sim\binomial{\ms}{1/2 + \dst}$ with $\dst\in[-1/2,1/2]$. Then, the indicator variable
    \begin{equation*}
        Y\defeq \indic{X\geq \ms/2}
    \end{equation*}
    is distributed as $Y\sim\bernoulli{1/2 + \beta}$ where
    \begin{equation*}
        \abs{\beta}=\bigOmega{\min(\sqrt{\ms}\dst, 1)}\,.
    \end{equation*}
\end{lemma}
We briefly describe how the above lemma may lead to~\cref{thm:prcoin:ldp:1bit:informal}. The first (by now somewhat standard) idea is to use a good error-correcting code such as the Hadamard code to define a family of $\ab$ sets\footnote{For simplicity of presentation, we implicitly assume in this overview that $\ab$ is a power of $2$, that $\ms$ is odd, and that everything is a multiple of what it needs to be.} $\chi_1,\dots,\chi_{\ab} \subseteq [\ab]$, each of size $\ab/2$, such that
\begin{equation}
    \label{eq:equation:hadamard:normpreserving}
\sum_{j=1}^{\ab} \Paren{\p(\chi_j) - \frac{1}{2}}^2 \geq \totalvardist{\p}{\uniform_{\ab}}^2
\end{equation}
(note that since each set has size $\ab/2$, if $\p=\uniform_{\ab}$ then the sum is $0$). Then we can partition the users into $\ab$ groups, where the users of group $j$ ``monitor'' $\chi_j$: by computing how many of their $\ms$ samples fall into their assigned set $\chi_j$. That is, each user of a given group $j$ now observes a random variable distributed as
\[
\binomial{\ms}{\p(\chi_j)}\,,
\]
where we can rewrite $\p(\chi_j) = 1/2+\dst_j$. By~\eqref{eq:equation:hadamard:normpreserving}, we then have that $\sum_{j=1}^{\ab} \dst_j^2$ is either $0$ (if $\p$ is uniform) or at least $\dst^2$ ($\p$ is far from it). Now, our user has a $\binomial{\ms}{\p(\chi_j)}$ in hand, and we would like them to send a single bit: this is where~\cref{thm:prcoin:ldp:1bit:informal} comes in, letting the user send the bit obtained by thresholding their observation at $\ms/2$. This yields a bit which has either bias $0$ (if $\p$ is uniform, since $\ms$ is odd) or \emph{some} bias $\beta_j$ such that $\beta_j = \bigOmega{\min(\sqrt{\ms}\dst_j, 1)}$. 

By piecing together (and centreing) the bits from one user of each of the $\ab$ groups, we can view it as one $\ab$-dimensional random bit vector in $\{-1,1\}^{\ab}$, with mean
\[
    (\beta_1,\dots,\beta_{\ab})
\]
This enables us to reduce our multi-sample-per-user uniformity testing problem to (privately) testing whether a \emph{product distribution} over $\{-1,+1\}^{\ab}$ is uniform or has \emph{mean vector} with norm at least
\begin{equation}
    \label{eq:lower:bound:norm:beta}
\sum_{j=1}^{\ab} \beta_j^2 \gtrsim \sum_{j=1}^{\ab} \min(\ms\dst_j^2, 1)
\end{equation}
The good news is that we now longer have to worry about user-level privacy: each user only needs to make their \emph{one} output bit $\priv$-LDP, which is easy to achieve via Randomized Response: by standard arguments, this only changes the mean testing problem by replacing the parameter $\beta^2 \eqdef \sum_{j=1}^{\ab} \beta_j^2$ by $O(\priv^2\beta^2)$. 
All that remains after that is to use an out-of-the-box (non-private) algorithm for mean testing of Rademacher-product distributions such as the ones in~\cite{CDKS:17,CKMUZ:19}, (not quite) establishing~\cref{thm:prcoin:ldp:1bit:informal}.

\paragraph{Are we there yet?} There are still, unfortunately, a few annoying issues with the above outline. The first is the $\min$ in~\eqref{eq:lower:bound:norm:beta}: we would \emph{like} to lower bound $\sum_{j=1}^{\ab} \beta_j^2$ by $\ms\sum_{j=1}^{\ab} \dst_j^2 \geq \ms\dst^2$, but this only allows us to get $\min(\ms\dst^2,1)$, which is (much) weaker for large $\ms$. Moreover, this is actually unavoidable in our reduction to mean testing of Rademacher-products!

The key insight is that this is only an issue when some of the sets $\chi_j$ have very large bias (high or low probability) under $\p$, in which case $\sqrt{\ms}\dst_j \gg 1$. We do not have control over this~--~it can happen!~--~but this should be an easy case to detect. And indeed, we can run, on a small number of users, an alternative protocol to detect whether there exists a $1\leq j^\ast\leq \ab$ with
$\dst_{j^\ast} = |\p(\chi_{j^\ast}) - 1/2| \gg 1/\sqrt{\ms}$. By standard concentration arguments for maximum of Binomials, losing only a $\log\ab$ factor, we can argue that, with high probability, this will not lead to erroneously rejecting the uniform distribution ($\p(\chi_j)=1/2$ for all $j$), but \emph{will} detect if any of these $\p(\chi_j)$ is overly biased. If this test passes, then whether $\p$ is uniform or not, we know that all $\dst_j$ are small enough, and so $\sum_{j=1}^{\ab} \beta_j^2 \gtrsim \frac{\ms\dst^2}{\log\ab}$.\medskip

This leads us to the second annoying issue: namely, that the above outline leads to a good sample complexity, but runs two distinct sub-protocols, one of them partitioning our $\ns$ users in $\ab$ distinct groups: this is very much an \emph{asymmetric} protocol, with users running $\ab+1$ distinct randomizers depending on their identity. To handle this, there is an ``obvious'', general-purpose solution: let each user pick uniformly at random which of these $\ab+1$ randomizers to run, and send both the (private) output and the (non-private) index of the randomizer they picked. By a coupon collector argument, this works with high probability, at the cost of a logarithmic factor in the number of users (and an $O(\log\ab)$ additional bits of communication per user).

As~\cref{thm:lower:bound:communication} asserts, the latter is necessary; the former, however, is very much wasteful. To avoid paying this logarithmic factor, we provide a generalization of the result for mean testing for Rademacher-product distributions, which allows for a different (and random) number of observations per coordinate and which we believe is of independent interest (\cref{thm:unifprodtest:symmetric}). Its analysis, while intuitively simple, is rather technical and relies on a symmetrization argument by negative association: we provide it in~\cref{sec:product:testing}.

Finally, we observe that while our focus is privacy, we obtain novel results for the \textit{bandwidth-constrained} setting~\cite{AcharyaLiuSun21} as well, where each user has limited communication budget with which to share their data.

\paragraph{Na\"ive approaches: ``why did you not simply do this?''}
When considering this model, the first and most natural approach is to consider the composition theorems of differential privacy. Simple composition of differential states that sending $T$ messages, each with $\priv$-differential privacy, results in a final privacy loss of $T\cdot\priv$. We use these theorems by pretending that there are $\ns\ms$ total users, each with $\ms=1$, and each guaranteed privacy $\priv'=\priv'(\priv,\ms)$. Even in the most generous case when our sample complexity is some $\ns\propto 1/\priv$, this results in no apparent gains from the additional samples over simply picking one of the $\ms$ uniformly at random. In reality the situation is worse, as we frequently have $\ns\propto 1/\priv^2$, meaning that sending $\ms$ messages, each with privacy parameter $\priv/\ms$ will result in an \emph{even worse} final sample complexity.

There also exists \emph{advanced composition}~\cite{dwork2010a,dwork2013} that say that the same composition satisfies \emph{approximate}\footnote{We do not further consider approximate-DP in this work, however one can informally think of $\privdelta$ as the ``probability of not being $\priv$-DP''.} $(\priv,\privdelta)$-differential privacy with final $\priv'=\priv\sqrt{2T\log(1/\privdelta)}+T\priv(e^\priv - 1)$. Applying this once again to some arbitrary problem with $\ns\propto1/\priv^2$, we see that our final sample complexity is again $\ns\ms\propto\ms/\priv^2$, yielding no gains.

As mentioned earlier, another simple approach exists when $\ms$ is large enough for each user to (non-privately) test on their own. In this regime we have each user test locally, reporting a bit that indicates \accept or \reject. We then simply learn the average bit using $\priv$-LDP, for which $\ns=\bigO{1/\priv^2}$ is sufficient. Therefore, for the very specific regime when $\ms> \sqrt{\ab}/\dst^2$ and $\ns>1/\priv^2$, we have a final sample complexity $\ns\ms=O(\sqrt{\ab}/(\dst^2\priv^2))$. The existence of this approach allows us to assume $\ms\leq \sqrt{\ab}/\dst^2$ for the remainder of this work.

\paragraph{A slightly less na\"ive approach (and why it is not enough).}

Another approach one could consider is one based on testing via repetition, used in~\cite[Section~2]{CanonneYang2024} for uniformity testing in the streaming setting. Assuming that $\ms\leq \ab$, we can run the following ``privatized'' version of this protocol. For each user, $i$, and each domain element $j \in [\ab]$, let $N_{i,j}$ the number of samples of user $i$ that are $j$. Then each user computes the statistic,
\[
    Z_i = \frac{1}{\ab}\sum\limits_{j=1}^{\ab}{\indic{N_{i,j}=0}}
\]
and applies an additive $\priv$-LDP noise mechanism. Then, the final statistic computed by the server would be the average of the private versions of all the $Z_i$'s, \ie the private version of 
$
    Z = \frac{1}{\ns}\sum_{i=1}^{\ns}{Z_i}\,.
$
To make each $Z_i$ private, each user would need to add noise to $Z_i$ that is calibrated to the sensitivity, $\Delta$, of $Z_i$ (\ie the effect of changing all the $m$ samples of that user on $Z_i$ in the worst case -- if $Z_i'$ is the new statistic after changing all of user $i$'s samples, then the sensitivity is the maximum possible value of $|Z_i - Z_i'|$). In this case, the sensitivity is $(\ms-1)/\ab$. The private version of $Z_i$ is then simply $
    \tilde{Z}_i = Z_i + \tau_i,
$
where $\tau_i \sim \Lap\left(\tfrac{\ms}{\priv}\right)$, and correspondingly the private version of $Z$ is 
$
    \tilde{Z} = \frac{1}{\ns}\sum\limits_{i=1}^{\ns}{\tilde{Z}_i} = Z + \frac{\tau}{\ns},
$
where $\tau = \tau_1 + \dots + \tau_\ns$.

Now, by linearity of expectations, $\expect{\tilde{Z}} = \expect{Z} + \expect{\tau/\ns} = \expect{Z}$. Then for $\p$ $\dst$-far from $\uniform$, following the proof of~\cite[Section~2]{CanonneYang2024} we get
\[
    \shortexpect_{\p}[\tilde{Z}] - \shortexpect_{\uniform}[\tilde{Z}] \geq \frac{\ms^2\dst^2}{4e\ab^2} \eqqcolon T
\]
and, since $\tau$ and $Z$ are independent,
$
    \var_{\p}(\tilde{Z}), \var_{\uniform}(\tilde{Z}) \leq \frac{2\ms^2}{\ns\ab^3} + \frac{2\Delta^2}{\ns\priv^2}\,.
$
By Chebyshev's inequality, we then get that a protocol which computes this $\tilde{Z}$ at the server and thresholds its value at $\shortexpect_{\uniform}[\tilde{Z}]+T/2$ gives a valid (private, private-coin) uniformity testing algorithm as long as 
\[
    \ns\ms^2 \gtrsim \frac{\ab}{\dst^4} + \frac{\ab^2}{\priv^2\dst^2}.
\]

Even though this is an improvement over the previously outlined methods, the sample complexity still falls short of what we are hoping for. One may hope to improve it by using techniques like the sensitivity-reducing mapping from~\cite{AliakbarpourDKR19,AliakbarpourBCR25} that could potentially reduce the sensitivity $\Delta$ from $\ms/\ab$ to $1/\ab$. However, even in the best possible case, that would only improve the sample complexity to
\[
    \ns\ms^2 \gtrsim \frac{\ab}{\dst^4}~~~\text{and}~~~ \ns\ms^4 \gtrsim \frac{\ab^2}{\priv^2\dst^2}.
\]
As we see, even this best-case scenario would not improve the dependence on $\ab$, and be very costly for small $\ms$.

\subsection{Prior work}
    \label{ssec:related}

Uniformity testing was first considered in the context of theoretical computer science by~\cite{GRexp:00}, and the optimal sample complexity $\bigTheta{\sqrt{\ab}/\dst^2}$ was obtained in~\cite{Paninski:08}. These results have since been generalized to the \emph{identity} testing problem, where the reference distribution is not uniform, a problem later proven to be formally equivalent to uniformity testing~\cite{DK:16, Goldreich:16, CanonneTopicsDT2022}. The task has since also been considered under (differential) privacy, first in the \emph{central} model of differential privacy~\cite{CDK:17}, where the tight sample complexity was later shown to be $\bigTheta{\sqrt{\ab}/\dst^2 + \sqrt{\ab}/(\dst\sqrt{\priv}) + \ab^{1/3}/(\dst^{4/3}\priv^{2/3}) + 1/\dst\priv}$~\cite{ASZ:18:DP,ADR:17}.

Under the more restrictive model of \emph{local} differential privacy, the question of testing was raised in~\cite{Sheffet:18} and has since been wholly resolved. It is known now that the tight sample complexity for non-interactive private-coin protocols is $\bigTheta{\ab^{3/2}/(\dst^2\priv^2)}$, while non-interactive public-coin protocols achieve $\bigTheta{\ab/(\dst^2\priv^2)}$~\cite{AcharyaCFST21,AcharyaCT:IT1,ACFT:19}. Furthermore, allowing interactivity cannot improve the sample complexity beyond the bound given in the public-coin setting~\cite{IIUIC,AJM:20,thomas_locally_2020}.

In the \emph{shuffle} model of differential privacy, where users' responses are permuted before being received by the central curator, an upper bound of $\bigO{({\ab^{3/4}\sqrt{\log(1/\delta)}})/{(\dst\priv)} + \sqrt{\ab}/\dst^2}$ was shown in~\cite{CanonneL22}. The best known lower bounds were derived from a connection to \emph{pan-privacy}~\cite{BCJM:20}.

The user-level setting is a natural generalization of any distributed statistical problem. In practice, it is restrictive to consider users holding only one sample from the distribution of interest whether we are testing, learning, or performing any other distributed task. For this reason, distribution \emph{learning} with multiple samples was studied under \emph{bandwidth constraints}, where an intricate interplay between the number of samples per-user and the number of bits available to communicate the samples was shown~\cite{AcharyaLiuSun21}. Under user-level differential privacy, these results were extended to show that (among other things) for small $\priv$, one achieves a sample complexity for learning of $\ns=\bigO{\ab^2/(\ms\dst^2\priv^2)}$. Comparing this to the case when $\ms=1$ (For example~\cite{ASZ:18:HR}) we can see that the risk decreases linearly with $\ms$. We could say that for the task of learning, having $\ms$ samples per-user is \emph{as beneficial as having $\ms$ times more users to query}.

Even though there is no direct reduction between the two settings, and (as inferred from the previous discussion) the situation, as $\ms$ grows, becomes quite subtle, this gives strong evidence that our results which show a sample complexity improvement by a factor $\ms$ are tight, at least in the most relevant parameter regime for $\ms$.\vspace{-1em}

\paragraph{Organization} 
The remainder of this paper is organized as follows. \Cref{sec:prelims} establishes necessary facts and definitions we use throughout the paper, including the necessary background on differential privacy, distribution testing, and some concentration inequalities important to this work. In~\cref{sec:publiccoin}, we present a result for the public-coin setting, combining known theorems from distribution testing and user-level private learning. \Cref{sec:noname:hadamard:protocol} introduces our main technical results in the form of non-private protocols that our private protocols are built upon. In particular, we show new algorithms for various testing problems, as well as several technical lemmata we believe are of independent interest. Due to the length and technical depth of this section, we defer the question of establishing privacy to~\cref{sec:privacy}, where we show how to adapt the results of the previous section to the user-level local differential privacy setting. Finally, we conclude with some open questions surrounding user-level local differential privacy. Wherever a proof, lemma, or algorithm is deferred from the main body, we include the appropriate reference to~\cref{app:deferred}.

\section{Preliminaries}
\label{sec:prelims}
\subsection{Notation}

We use $\ns$ to refer to the number of users participating in the protocol (or in the dataset) and $\ms$ to refer to the number of samples (or data points) each user has. Additionally, we use $\ab$ to denote the domain size, $\priv$ to denote the privacy parameter, $\dst$ to denote the accuracy parameter, and $\errprob$ to denote the probability of failure (in terms of accuracy) of our protocol. We also write $[\ab] \eqdef \{1,2,\dots,\ab\}$.

Next, for a distribution $\p$, $\p^{\otimes\ms}$ refers to the $m$-dimensional product distribution with each marginal $\p$. For two distributions $\p$ and $\q$ over the same (countably infinite) domain $\Omega$, we denote by $\dtv(\p,\p)$ the total variation distance between them, defined as
\[
\totalvardist{\p}{\q} = \sup_{S\subseteq \Omega}(\p(S)-\q(S)) = \frac{1}{2}\sum_{x\in\Omega} \abs{\p(x)-\q(x)}\,.
\]
In the context of Bernoulli and Binomial distributions with parameter $p$, we call their deviation from $1/2$ (\ie $p-1/2$) their \emph{bias}.

Finally, we use the notation $\gtrsim$, $\lesssim$, and $\asymp$ to denote the (sometimes slightly more convenient) analogues of the $\Omega(\cdot)$,  $O(\cdot)$,  and $\Theta(\cdot)$ notation: specifically, for two sequences $(a_n)_n$, $(b_n)_n$ indexed by some parameter $n$, we write $a_n\lesssim b_n$ if there exists $C>0$ such that $a_n \leq C\cdot b_n$ for every $n \geq 0$, with the inequality reversed for $\gtrsim$. $a_n \asymp b_n$ then denotes that both $a_n \lesssim b_n$  and $a_n \gtrsim b_n$ hold. Throughout, $\land$ and $\lor$ denote minimum and maximum: $a\land b = \min(a,b)$, and $a\lor b = \max(a,b)$.

\subsection{Differential Privacy}

We first provide the necessary notions and results from the differential privacy (DP), starting with the definition of \emph{local} differential privacy (LDP).
\begin{definition}[Local Differential Privacy~\cite{KLNRS:11,Warner:65}]
\label{def:ldp}
    For $\priv>0$, a randomized algorithm $\pmech\colon\dmain\to\range$ provides $\priv$-Local Differential Privacy if for all $i,j\in\dmain$
    \begin{equation*}
        \max_{y\in\range}\frac{\Pr[Q(i)=y]}{\Pr[Q(j)=y]}\leq e^\priv.
    \end{equation*}
\end{definition}
While Local Differential Privacy is somewhat involved to define for interactive protocols, where each user can send (in an adaptive manner) several messages, it is simpler in our setting. We consider non-interactive protocols, where each user only sends one message to the server. When each user holds several datapoints (that is, $\cX = [\ab]^{\ms}$), the above definition then directly corresponds to the \emph{user-level} LDP guarantee considered in this paper.

We will rely on the (binary) \emph{Randomized Response} mechanism, which is an optimal $\priv$-local differential privacy protocol for 1-bit inputs~\cite{Kasiviswanathan11}. Formally, let $Q$ be the local randomizer for this protocol where $Q(x)$ is a random variable, and $Q(y\,\vert\, x)$ is the probability of seeing output $y$ on input $x$.
\[
Q(y\,\vert\, x)=\begin{cases}
    \frac{e^\priv}{e^\priv + 1} &y=x\\
    \frac{1}{e^\priv + 1}& y\neq x.
\end{cases}
\]

\subsection{Distribution and Uniformity Testing}

Here, we formally state the problem and setting. We are interested in the problem of \emph{uniformity testing} with multiple samples per user. Specifically, given a discrete distribution $\p$ over $\ab$ elements, which we assume without loss of generality is over the domain $[\ab]$, each user $i=1,\dots,\ns$ holds a multi-sample $X_i \sim \p^{\otimes\ms}$ which is an $\ms$-dimensional vector of \iid samples from $\p$. We are interested in algorithms that can distinguish between the cases where $\dtv(\p,\uniform)=0$ and $\dtv(\p,\uniform)>\dst$ via the fewest possible number of users $\ns$. As mentioned in the introduction, barring some kind of \emph{information constraint} per-user, this problem reduces to the well-studied case where $\ms=1$ by having each user send all $\ms$ of their samples and treating each as its own message. This problem becomes non-trivial when we are either \emph{bandwidth-constrained} or \emph{privacy-constrained}. In the former, users have up to $\numbits$ bits with which to communicate their samples. In the latter, the user is constrained by \emph{differential privacy}. 

In the user-level LDP setting, upon receiving $\ns$ private responses $Y_1,\dots,Y_\ns$ from the users, the server must distinguish between the two cases:
\begin{itemize}
    \item $\dtv(\p,\uniform)=0$, and
    \item $\dtv(\p,\uniform)>\dst$,
\end{itemize}
with probability at least $2/3$, while satisfying $\priv$-user-level LDP. (The threshold $2/3$ is somewhat arbitrary, and can be amplified to any $1-\errprob$ by standard arguments at a sample complexity cost of $O(\log(1/\errprob))$.)
\subsection{Useful Probability Tools}

We first recall Cantelli's inequality, a one-sided version of Chebyshev's inequality:
\begin{lemma}[Cantelli's inequality]
    \label{lemma:cantelli}
    Let $X$ be a real-valued random variable with finite variance. Then, for every $\lambda > 0$,
    \[
    \probaOf{X \geq \expect{X} + \lambda }, \probaOf{X \leq \expect{X} - \lambda } \leq \frac{\var[X]}{\var[X]+\lambda^2}\,.
    \]
\end{lemma}

We will also require the following standard tail bound for subgaussian random variables (see \eg~\cite{MathSE23}):
\begin{lemma}
    \label{lemma:maxsubgaussian}
Let $X_1,\dots,X_n$ be (not necessarily independent) $\sigma^2$-subgaussian random variables with mean zero. Then
\[
\mathbb{E}[\max_{1\leq i\leq n} X_i] \leq \sqrt{2\sigma^2\log n} 
\]
and, for every $t>0$, 
\[
\mathbb{P}\!\left\{\max_{1\leq i\leq n} X_i \geq \sqrt{2\sigma^2(\log n + t)}\right\} \leq e^{-t}\,. 
\]
In particular, this implies that, for every $t>0$, 
\[
\mathbb{P}\!\left\{\max_{1\leq i\leq n} |X_i| \geq \sqrt{2\sigma^2(\log(2n) + t)}\right\} \leq e^{-t}\,. 
\]
\end{lemma}

\section{Public-Coin via Domain Compression}
\label{sec:publiccoin}
In this section, as a warmup, we show how to establish our public-coin result, \cref{thm:pubcoin:ldp:1bit:intro}, leveraging existing algorithm building blocks, \emph{domain compression} and \emph{user-level coin estimation}. We restate it below:
\begin{theorem}[LDP 1-bit asymmetric public coin uniformity testing]
\label{thm:pubcoin:ldp:1bit}
    There exists an \emph{asymmetric, public-coin, user-level} locally differentially private algorithm for uniformity testing (over domain of size $\ab$) which on privacy and distance parameters $\priv>0$ and $\dst\in(0,1]$ takes
\begin{equation*}
    \ns = \bigO{\frac{\ab}{\ms\dst^2\priv^2}}
\end{equation*}
users, each holding $\ms$ \iid samples from the unknown distribution and sending \emph{one} bit of communication.
\end{theorem}
\begin{proof}[Proof (Detailed sketch)]
    The protocol works as follows: using public randomness, the $\ns$ users jointly perform domain compression, \ie~a type of hashing of the domain, reducing the domain size from $\ab$ to $2$. By the following lemma, with constant probability, this preserves the distance between distributions up to a $\sqrt{\ab}$ factor:

    \begin{lemma}[Domain Compression~\cite{CanonneTopicsDT2022}]
There exists absolute constants $c_1,c_2>0$ such that the following holds. For any $2\leq\ell\leq\ab$ and any $\p,\q\in\Delta_{\ab}$,
\begin{equation*}
    \probaDistrOf{\Pi}{\dtv(\p_\Pi,\q_\Pi)\geq c_1 \sqrt{\frac{\ell}{\ab}}\dtv(\p,\q)}\geq c_2
\end{equation*}
where $\Pi=(\Pi_1,\dots,\Pi_\ell)$ is a uniformly random partition of $[\ab]$ into $\ell$ subsets, and $\p_\Pi\in\Delta_\ell$ denotes the probability distribution on $[\ell]$ induced by $\p$ and $\Pi$.
\end{lemma}

That is, if $\p=\uniform_{\ab}$, then this (always) results on a distribution uniform over a domain size $2$ (\ie~a fair coin), while if $\p$ was $\dst$-far from uniform this results (with probability $\Omega(1)$ in a coin with bias $\dst' = \Omega(\dst/\sqrt{\ab})$.

All $\ns$ users now having $\ms$ samples from the same induced ``coin'', all that remains is to learn the bias of this coin to accuracy $\dst'/4$ in order to distinguish between the two cases. This can be done with the following existing protocol:
\begin{theorem}[{\cite[Theorem~3.2]{AcharyaLiuSun21}}]
    \label{theo:personlevel:biasedcoin}
For $\priv \in(0,1]$, there exists a private-coin algorithm for person-level coin bias estimation with
\[
\expect{(\hat{p}-p)^2} = \bigO{\frac{1}{\ns\ms\priv^2}}
\]
assuming $\ns \geq C\cdot \log(\ms)/\priv^2$, where $C>0$ is an absolute constant.
\end{theorem}
Phrased differently, with the above one can privately learn the bias of a coin to an additive $\dst'$, with arbitrary (constant) probability, in the user-level LDP setting, using $\ns  = O(1/(\ms\dst^2\priv^2)$ users.

Putting things together and recalling our setting of $\dst'$, the above protocol then succeeds with (small) constant probability, as long as
    \[
        \ns\ms \gtrsim \frac{\ab}{\dst^2\priv^2}
    \]
``as claimed.'' This is so far a symmetric protocol: but the probability of success, for distributions far from uniform, is quite small, as the domain compression only preserves the distances with small probability. To amplify the probability of success to $2/3$ by standard techniques, we repeat the protocol on disjoint batches of users and combine their answers via a majority vote, leading to an asymmetric protocol. Finally, note that the condition on $\ns$ from \cref{theo:personlevel:biasedcoin} is indeed satisfied, as $\frac{\ab}{\ms\dst^2\priv^2} \gg \frac{\log(\ms)}{\priv^2}$.
\end{proof}

\section{Low-Bandwidth Private-Coin via Hadamard Matrices}
\label{sec:noname:hadamard:protocol}
As discussed in~\cref{sec:introduction}, our high-level approach will proceed as follows: given $\ns$ users, each holding $\ms$ samples from some unknown discrete distribution $\p$ over $\ab=2^t$ elements, we first assign each user $i$ a group $G_{j}$, for $2\leq j\leq \ab$ (intentionally dropping a group  so that we have $\ab-1$ groups). As a first attempt, consider this assignment to be a deterministic function $j(i)$ that evenly partitions the $\ns$ users among $\ab$ groups. Later, to remove this coordination step, which would lead to an asymmetric protocol, we will instead have each users select the index $j$ of their group uniformly at random. We describe our algorithm in stages: first, we give in~\cref{ssec:asymmetric:smallm} an asymmetric algorithm well-suited for ``small'' values of $\ms$, before providing a complementary approach to detect large variations from uniformity (\cref{sec:complementary:cornercase}), and explaining how to put them together in~\cref{ssec:combining:nonprivate}.

\subsection{An algorithm for small \texorpdfstring{$\ms$}{m}}
    \label{ssec:asymmetric:smallm}

Using the same indexing as for the groups, let $\chi_j$ be the set of indices where a 1 appears in the column $j$ of the Hadamard matrix $H\in \bool^{\ab\times\ab}$, that is, 

\begin{equation}
\label{eq:hadamard:subsets}
    \chi_j \eqdef \setOfSuchThat{r \in [\ab]}{ H_{rj} = +1}\,, \qquad 2\leq j\leq \ab.
\end{equation}

The properties of Hadamard matrices ensure that each of its columns is a column vector of length $2^t$, where half of the positions are $+1$, and the other half are $-1$, and so $\abs{\chi_j} = \ab/2$ for all $j$. Hereafter, we will write $\p(\chi_j)$ to denote the probability that a sample under $\p$ falls in set $\chi_j$, \ie $\p(\chi_j)=\sum_{r\in\chi_j}\p_r$. User $i$ then computes 
\begin{equation*}
    X_i=\sum_{\ell=1}^\ms \indic{x_\ell\in\chi_{j(i)}}, 
\end{equation*}
\ie the number (out of their $\ms$ samples) that lie within $\chi_{j}$ (the subset of the domain they are ``monitoring'').
One can observe the following distinction of cases:
\begin{itemize}
    \item If $\p=\uniform_{\ab}$, then $X_i\sim\binomial{\ms}{1/2}$.
    \item If $\dtv(\p,\uniform_{\ab})>\dst$, then $X_i\sim\binomial{\ms}{1/2+\dst_j}$, for some $\dst_j$ which will be related to $\dst$ by~\cref{lemma:hadamard:norm:preservation}.
\end{itemize}
We call $\dst_j$ the \emph{bias} observed by group $G_j$. User $i$, instead of directly sending $X_i$, will then compute the one-bit indicator
\[
Y_i =\indic{X_i\geq \frac{\ms+1}{2}},
\]
and send this to the server. We will show that, given $\ns=\bigO{\frac{\ab^{3/2}}{\ms\dst^2}}$ independently drawn samples $Y_1,\dots, Y_\ns$ as above, there exists an algorithm that distinguishes between $\p=\uniform_{\ab}$ and $\p$ $\dst$-far from $\uniform_{\ab}$ (with probability at least $2/3$).

\begin{theorem}[Asymmetric 1-bit multi-sample uniformity tester]
\label{thm:prcoin:1bit:assym}
    Given $\ns$ users, each holding $\ms$ samples of some unknown distribution $\p$ on $\ab$ elements. There exists an algorithm (\cref{algo:private:coin:asymmetric:nonprivate}) that distinguishes between $\p=\uniform$ and $\dtv(\p,\uniform)>\dst$ using
    \[
    \ns = \bigO{\frac{\ab^{3/2}}{\ms\dst^2\wedge 1}},
    \]
    samples. Moreover, each user only sends \emph{one} bit.
\end{theorem}
(Note that this protocol, which follows the above outline, is asymmetric, as users are partitioned in $\ab$ distinct groups of equal size, and users from different groups process their $\ms$ samples differently.) This first algorithm is particularly well suited to the ``small $\ms$ regime'' where $\ms \leq 1/\dst^2$, where it achieves the desired number of users.

The rest of this subsection is dedicated to establishing~\cref{thm:prcoin:1bit:assym}. In view of this, we will need three ingredients: first, demonstrating that $Y_i$ has bias $\sqrt{\ms}\dst_{j}$ (this factor $\sqrt{\ms}$ is where the dependence on $\ms$ in the final bound will come from); second, showing that a distribution that is $\dst$-far from uniform induces bias in each group, such that $\sum_{j=1}^{\ab} \dst_j^2\approx\dst^2$; and third, a known result on uniformity testing for Rademacher-product distributions. 
The proof proceeds immediately from these results: Each user sends their single bit $Y_i$ with bias $\sqrt{\ms}\dst_j$. We then take one bit from each group and arrange them into a vector. This vector simulates a sample from the product distribution with mean $\mu=(\p(\chi_1),\dots,\p(\chi_{\ab}))$. Applying~\cref{lemma:hadamard:norm:preservation} we get that the expected $\lp[2]^2$ norm of this vector is at least $\ms\dst^2$. Plugging in an $\lp[2]^2$-norm product tester as a black box, we conclude our proof.

As stated above, we begin by capturing the behaviour of $Y_i$ as a function of each $\dst_j$. Recall that $Y_i$ is an indicator variable that declares whether a binomial $X_i$ exceeded the mean it ``should'' have in the uniform case: in this sense, going from $X_i$ to $Y_i$ converts a ``many-bit'' sample to a ``single-bit'' one. We show that for $X\sim\binomial{\ms}{1/2 \pm \alpha}$ with small bias $\alpha$, this indicator behaves as a Bernoulli with mean $\sqrt{\ms}\alpha$, and so this conversion from many bits to one still preserves both bias $\alpha$ and a dependence on $\ms$. For ease of exposition, we defer its proof to the end of the section. 
\begin{algorithm}[htbp!]
    \caption{Asymmetric Hadamard Protocol for Uniformity Testing}
    \label{alg:hadamard-main}
    \begin{algorithmic}
        \Require $\ns$ users, each with vector $\vec{x}_i\in[\ab]^\ms$ holding $\ms$ samples from distribution $\p$ over $[\ab]$; Distance parameter $\dst$; Hadamard matrix $H$ of size $\ab \times \ab$
        \Ensure
        
        \State $\chi_j \gets \{r \in [\ab] : H_{rj} = +1\}$ for each column $j \in [\ab]$
        \State $N\gets \frac{\ns}{\ab-1}$ \Comment{Number of users in each group}
        \State Partition users into $\ab-1$ groups $G_2, \ldots, G_{\ab}$ of size $N$ each
        
        \For{$j = 2$ to $\ab$} 
            \For{each user $i \in G_j$}
                \State $X_i \gets \sum_{\ell=1}^{\ms} \indic{x_{i,\ell} \in \chi_j}$ \Comment{Count samples in monitored subset}
                \State $Y_i \gets \indic{X_i \geq \frac{\ms+1}{2}}$ \Comment{Threshold to single bit}
                \State \Return $Y_i$
            \EndFor
        \EndFor
        
        \State \textbf{Server:}
        \State $Z_i \gets 2(Y_i - 1/2)$ for all $i \in [\ns]$ \Comment{Convert to Rademacher}
        \For{$\ell=1$ to $N$}
            \State Initialize $\vec{Z}^{(\ell)}\gets(0,\dots,0)$
            \For{each group $j=2$ to $\ab$}
                \State $\vec{Z}^{(\ell)}_j\gets$ next unused $Z_i\in G_j$
            \EndFor
        \EndFor

        \State Run product distribution uniformity test on $\{\vec{Z}^{(1)}, \ldots, \vec{Z}^{(N)}\}$ with parameter $\gamma \gets \sqrt{\ms}\dst \wedge 1$
        \State \Return test result
    \end{algorithmic}
    \label{algo:private:coin:asymmetric:nonprivate}
\end{algorithm}

\begin{lemma}[From Many Bits, One] %
    \label{lemma:epluribusunum}
    Let $m = 2\ell-1 \geq 1$ be an odd integer, and $\alpha \in [-1/2,1/2]$. Define
    \[
        Y \eqdef \indic{X \geq \ell }
    \]
    where $X\sim \binomial{m}{1/2+\alpha}$. Then $Y\sim \bernoulli{1/2+\beta}$, where
    \[
        \abs{\beta} = \bigOmega{\min(\sqrt{m}\abs{\alpha},1)}\,.
    \]
    (Moreover, if $\alpha=0$, then $\beta=0$.)
\end{lemma}

While the bias $\dst_j$ observed by each group is clearly distribution-dependent, we need to relate them explicitly to the distance parameter $\dst$. To do so, we use the known fact that multiplication by a Hadamard matrix is $\lp[2]$-norm preserving. 
\begin{lemma}[Hadamard transform is norm-preserving]
    \label{lemma:hadamard:norm:preservation}
    As in the process described above, let $\p$ be a distribution with $\totalvardist{\p}{\uniform}\geq \alpha$. Let $H$ be the $2^t \times 2^t$ Hadamard matrix with $\pm 1$ entries. For each column $j \in \{2, 3, \ldots, 2^t\}$ (excluding the all-ones column), define $\chi_j \eqdef \setOfSuchThat{i \in [\ab]}{ H_{ij} = +1}$, and let $\dst_j \eqdef \p(\chi_j) - \frac{1}{2}$ be the bias of column $j$ under distribution $\p$. Then
\[
    \sum_{j=1}^{\ab} \dst_j^2 \geq \dst^2.
\]
\end{lemma}
\noindent(This follows, for instance, from~\cite[Lemma~3]{AcharyaCFST21}, combined with Cauchy--Schwarz to relate total variation and $\lp[2]$ distances.) For completeness, we provide a self-contained proof in~\cref{app:hadnormpreserved}.\smallskip

Finally, the third (and last) ingredient missing is an algorithm to test, given \iid observations from a product distribution on $\ab$ bits, whether the mean vector is zero (uniform distribution) or has large norm:
\begin{lemma}[Uniformity testing on product distributions]
\label{thm:unifprodtest}
Fix any $\dims \geq 2$. 
There exists an algorithm that, given a parameter $\gamma\in(0,\sqrt{\dims}]$ and $\ns$ \iid samples from a product distribution $\p$ on $\{-1,1\}^\dims$ with $\mu\eqdef\shortexpect_{X\sim \p}[X]\in[-1,1]^\dims$, has the following guarantees.
\begin{itemize}
    \item If $\normtwo{\mu}\leq \frac{\gamma}{2}$, the algorithm returns $\accept$ w.p. $\geq \frac{2}{3}$;
    \item If $\normtwo{\mu}\geq \gamma$, the algorithm returns $\reject$ w.p. $\geq \frac{2}{3}$;
\end{itemize}
as long as $\ns\geq C\frac{\sqrt{\dims}}{\gamma^2}$, for some absolute constant $C>0$. 
\end{lemma}
\noindent For a proof of this in the case $\gamma \in (0,1]$, see, \eg~\cite[Section~2.1]{CDKS:17} or~\cite[Lemma~4.2]{CKMUZ:19}, which establish this along the way, while focusing on testing in total variation distance, or~\cite[Theorem~4.1]{CanonneCKLW21}, which provides slightly stronger guarantees. We note that while stated only for $\gamma \in (0,1]$, the proofs above actually implicitly show the result for the whole range of $\gamma$.
For completeness, we provide a self-contained proof (for the whole range of $\gamma$) in~\cref{sec:product:testing}.

With these three building blocks in hand, we are ready to analyze~\cref{algo:private:coin:asymmetric:nonprivate}:
\begin{proofof}{\cref{thm:prcoin:1bit:assym}}
Recall that user $i$ is deterministically assigned group index $j(i)\in\{2,\dots,\ab\}$ such that each group is of equal size.\footnote{The first column (and row) of the Hadamard matrix are all-ones, and can be ignored, or otherwise simulated if necessary, as any user's behaviour in this group would be to deterministically send a 1.} Where $i$ is clear from context or not relevant we suppress this notation, letting $j\defeq j(i)$. We have each user send their $Y_i$ per~\cref{alg:hadamard-main} and focus on the server's view. Clearly $Y_i$ is distributed as some (yet unknown) Bernoulli $Y_i\sim\bernoulli{\frac{1}{2}+\beta_j}$, where $\beta_j$ depends on $\p$. 
Centering each of these, we get the Rademacher random variables
\[
Z_i = 2(Y_i - 1/2) \in \{-1,1\} \tag{$1\leq i\leq \ns$}
\]
each with mean $\expect{Z_i}=\beta_{j(i)}$. Now, we wish to apply~\cref{thm:unifprodtest} which takes as input samples from some product distribution. To facilitate this we construct our own vector samples by concatenating one $Z_i$ from each group. Each group $G_j$ contains $N=\ns/(\ab-1)$ users. So, as described in~\cref{alg:hadamard-main}, we create the vectors $\{\vec{Z}^{(1)},\dots,\vec{Z}^{(N)}\}$ so that $\vec{Z}^{(\ell)}_j$ holds the bit sent by the $\ell$'th user of group $j$. 

Each $\vec{Z}^{(\ell)}$ therefore has mean vector $\mu=(\beta_1,\dots,\beta_{\ab})$. Applying~\cref{lemma:epluribusunum}, we get that, for all $1\leq j\leq \ab$, (1) if $\p=\uniform_{\ab}$, then $\beta_j = 0$; and (2) otherwise, we have $\abs{\beta_j}=\Omega(\min(\sqrt{\ms}\dst_j, 1))$. Combining this with~\cref{lemma:hadamard:norm:preservation}, this vector $\mu$ satisfies:
\begin{itemize}
    \item if $\p=\uniform$, then $\mu = \mathbf{0}^{\ab}$;
    \item if $\totalvardist{\p}{\uniform} > \dst$, then
    \begin{equation}
        \label{eq:lb:norm:mean:product}
        \normtwo{\mu}^2 
        = \sum_{j=1}^{\ab} \beta_j^2 
        \gtrsim \sum_{j=1}^{\ab} \Paren{\ms\dst_{j}^2\land 1} 
    \end{equation}
\end{itemize}
and the RHS is at least $\ms \dst^2 \land 1$.
This is exactly the setting we need to invoke~\cref{thm:unifprodtest}: setting $\gamma=\sqrt{\ms}\dst\land 1$, and $N=\ns/(\ab-1)$, this yields a final sample complexity of $\ns = \bigO{\frac{\ab^{3/2}}{\ms\dst^2 \land 1}}$.
\end{proofof}

Before proceeding to the next component of our algorithm, it remains to establish~\cref{lemma:epluribusunum}.

\begin{proofof}{\cref{lemma:epluribusunum}}
    We start with the ``Moreover'' statement: if $\alpha=0$, then $X\sim \binomial{m}{1/2}$. By symmetry, $m-X\sim \binomial{m}{1/2}$, and since $m=2\ell-1$,
    \[
        \probaOf{X \geq \ell } 
        = \probaOf{m-X \geq \ell }
        = \probaOf{X \leq m-\ell }
        = \probaOf{X \leq \ell -1 }
    \]
    from which $\probaOf{X \geq \ell }=1/2$.

    Assume without loss of generality that $\alpha \geq 0$ (as otherwise we can consider $m - X$ instead). To establish the first part of the statement, we will distinguish between three cases, depending on how large $\alpha$.
    \begin{itemize}
        \item First case: $\alpha \geq 2/\sqrt{m}$. By Cantelli's inequality (\cref{lemma:cantelli}), since $\expect{X} = \ell - 1/2 + m\alpha$,
        \[
            \probaOf{ X < \ell }
            = \probaOf{ X < \expect{X} - \Paren{m\alpha - \frac{1}{2}} }
            \leq \frac{\var[X]}{\var[X]+\Paren{m\alpha - \frac{1}{2}}^2}
            \leq \frac{1}{1+m\alpha^2} \leq \frac{1}{5}
        \]
        using $m\alpha > 1$ and $\var[X] \leq m/4$. This shows that $\probaOf{ X \geq \ell } \geq 4/5$, \ie  $\beta \geq 3/10$.
        \item Second case: $\alpha < 3/m$.\footnote{The choice of the constant $3$ in $3/m$ (instead of the more natural $1/m$) may appear somewhat arbitrary: this specific value for the cut-off will be useful, for technical reasons, in the third and last case.} In this case, the mean of $X$ only differs by $\alpha m < 3 = O(1)$ from $m/2$, the mean of a standard Binomial $\tilde{X} \sim \binomial{m}{1/2}$ (and the modes of the two distributions are either the same or very close integers), so the change in probability mass between the two is quite subtle. We can provide a coupling between $X$ and $\tilde{X}$ as follows:
        \[
                X = \tilde{X} + \tilde{Z}
        \]
        where the distribution of $\tilde{Z}$, conditioned on $\tilde{X}$, is $\tilde{Z}\sim \binomial{\tilde{X}}{2\alpha}$. One can check that this satisfies $X\sim\binomial{m}{1/2+\alpha}$ (\ie this is a valid coupling) and directly implies that $X \geq \tilde{X}$ a.s. Now, we have
        \begin{align*}
            \probaOf{ X \geq \ell} 
            &\geq \probaOf{ \tilde{X} \geq \ell } + \probaOf{\tilde{X}=\ell-1, \tilde{Z} \geq 1} \\
            &= \frac{1}{2} + \probaOf{\tilde{X}=\ell-1}\probaOf{\binomial{\ell-1}{2\alpha} \geq 1} \\
            &= \frac{1}{2} + \frac{\binom{2\ell-1}{\ell-1}}{2^m}\cdot \Paren{1-(1-2\alpha)^{\ell-1}} \\
            &= \frac{1}{2} + \bigTheta{ \frac{1}{\sqrt{\ell}}\cdot \ell\alpha }\,,
        \end{align*}
        using in the last step that $\alpha = O(1/\ell)$. This shows that in this case $\beta = \bigTheta{\sqrt{\ell}\alpha} = \bigTheta{\sqrt{m}\alpha}$.
        \item Third case: $3/m \leq \alpha < 2/\sqrt{m}$. In this regime, we can rely on the Gaussian approximation, as quantified by the Berry--Esseen theorem (see, \eg~\cite[Section~11.5]{ODonnell14}, which guarantees that the CDF $F$ of the normalized version of $X$,
        \[
            X' \eqdef \frac{X-\expect{X}}{\sqrt{\var[X]}} = \frac{2X-m(1+2\alpha)}{\sqrt{m(1-4\alpha^2})}
        \]
        is pointwise close to the CDF $\Phi$ of a standard Gaussian $Z \sim \gaussian{0}{1}$:
        \[
            \sup_{x\in \R} \abs{F(x) - \Phi(x)} \leq \frac{C}{\sqrt{m}}\,,
        \]
        for some absolute constant $C>0$ (one can take $C = 0.56$). In particular, this implies, in our case, that
        \begin{align*}
            \probaOf{X < \ell} &=  \probaOf{X' < -\frac{2m\alpha-1}{\sqrt{m(1-4\alpha^2)}}} \\
            &\leq \probaOf{X' < -\frac{\sqrt{m}\alpha}{\sqrt{1-4\alpha^2}}} \tag{as $2m\alpha - 1 \geq m\alpha$} \\
            &\leq \probaOf{X' < -\sqrt{m}\alpha}  \\
            &\leq \probaOf{ Z < -\sqrt{m}\alpha} + \frac{C}{\sqrt{m}} \tag{Berry--Esseen}\\
            &\leq \frac{1}{2} - \frac{\sqrt{m}\alpha}{5} + \frac{C}{\sqrt{m}} \tag{Studying $\Phi$ for $\sqrt{m}\alpha\in[0,2]$}\\
            &\leq \frac{1}{2} - \frac{1}{100}\sqrt{m}\alpha
        \end{align*}
        the last step using that $m\alpha \geq 3$ and $C\leq 0.56$. This shows that, in this regime as well, $\beta = \bigOmega{\sqrt{m}\alpha}$.
    \end{itemize}
    This concludes the distinction of cases, and the proof.
\end{proofof}

\subsection{An algorithm for large \texorpdfstring{$\ms$}{m}}
\label{sec:complementary:cornercase}
The above approach gives the ``right'' sample complexity under the restriction that 
\[
\ms \dst_j^2 \leq 1, \qquad \forall j \in [\ab]\,,
\]
where $\dst_j = \p(\chi_j) - \frac{1}{2}$. We here provide a different protocol, which works well when at least one of the $|\dst_j|$'s is large. The main idea is to have each user just check is \emph{any} of the $\ab$ sets $\chi_j$ receives a lot more (or less) than half of their $\ms$ samples, namely, $\frac{\ms}{2}\pm \Omega(T)$ for some suitable threshold $T = \bigTheta{\sqrt{\ms\log(\ns\ab)}}$. If $\p$ is uniform, then this is highly unlikely, as by~\cref{lemma:maxsubgaussian} the maximum deviation of $\ab$ subgaussian random variables (here, Binomials) from their mean is less than $T$ with overwhelming probability. But if $\p$ is not uniform \emph{and} one of the $|\dst_j|$'s is large, then the number of samples falling in the corresponding set $\chi_j$ is a very biased Binomial, and for every user this set will receive $\frac{\ms}{2}\pm \Omega(T)$ samples with high constant probability. We formalize this idea in~\cref{alg:large:m:tester}, starting with the non-private version; and prove the following theorem:
\begin{algorithm}[htbp!]
    \caption{Symmetric protocol for large $\ms$}
    \label{alg:large:m:tester}
    \begin{algorithmic}
        \Require $\ns$ users, each with vector $\vec{x}_i\in[\ab]^\ms$ holding $\ms$ samples from distribution $\p$ over $[\ab]$; Hadamard matrix $H$ of size $\ab \times \ab$.
        \Ensure With high probability detect when $\exists j^*\dst_{j^*}>\Omega(\sqrt{\log(\ns\ab)/\ms}$.

        \State $\chi_j \gets \{r \in [\ab] : H_{rj} = +1\}$ for each column $j \in [\ab]$
        \State $T\gets \frac{1}{2}\sqrt{\ms\ln(20\ns\ab)}$
        \For{each user $i\in[\ns]$}
            \For{each $j=2$ to $\ab$}
                \State $V_j^{(i)}\gets \sum_{\ell=1}^{\ms} \indic{x_{i,\ell} \in \chi_j}$ \Comment{Count samples in each subset}
            \EndFor
            \If{$\max_{2\leq j\leq \ab}\abs{V_j^{(i)}-\frac{\ms}{2}}>T$}
                \State Set $v_i \gets 1$
            \Else
                \State Set $v_i \gets 0$
            \EndIf
        \EndFor
        \State \textbf{Server:}
        \If{$\sum_{i=1}^\ns v_i \geq \frac{\ns}{2}$}  \Comment{Majority vote}
            \State \Return \reject
        \Else
            \State \Return \accept
        \EndIf

    \end{algorithmic}
\end{algorithm}

%
%
\begin{theorem}
    \label{lemma:complementary}
    There exists an algorithm (\cref{alg:large:m:tester}) with the following guarantees:
    \begin{itemize}
        \item If $\p=\uniform_{\ab}$, then the center outputs \accept with probability at least $9/10$;
        \item If there exists $1\leq j^\ast\leq \ab$ such that $\dst_{j^\ast} = \bigOmega{\sqrt{\log(\ns\ab)/\ms}}$, then the center outputs \reject with probability at least $9/10$.
    \end{itemize}
    Moreover, each person sends only one bit.
\end{theorem}
\begin{proof} In what follows, we set $R\eqdef 2\frac{T}{\sqrt{\ms}} = \sqrt{\ln(20\ns\ab)}$.
\begin{itemize}
\item If $\p$ is uniform, then $\p(\chi_j) = 1/2$ for all $j$. By standard tail bounds for subgaussian random variables, this means that, for every $1\leq i\leq \ns$
\[
    \probaOf{\max_{1\leq j\leq \ab} |V^{(i)}_j - \tfrac{\ms}{2}| > T} \leq \frac{1}{10\ns}
\]
given our setting of $T$ (specifically, we apply \cref{lemma:maxsubgaussian} with $t=\ln(10\ns)$ to the $\ab$ mean-zero random variables $(V^{(i)}_j - \tfrac{\ms}{2})_j$, which are centered Binomials, and thus $\frac{\ms}{4}$-subgaussian). By a union bound, we get $\probaOf{\exists i, v_i=1} \leq 1/10$, and the server outputs \accept with probability at least $9/10$. 

\item Turning to the other case, assume there exists $j^\ast\in[\ab]$ such that $|\dst_{j^\ast}| > \frac{R}{\sqrt{\ms}}$. Then, for every person $i$,
$\abs{\expect{V^{(i)}_{j^\ast}} - \frac{\ms}{2}} > R\sqrt{\ms}$, and 
\[
\max_{1\leq j\leq \ab} |V^{(i)}_j - \frac{\ms}{2}|
\geq |V^{(i)}_{j^\ast} - \frac{\ms}{2}| > T,
\]
where the last inequality holds with probability at least $2/3$ by Chebyshev (as $R \geq \sqrt{\ln(20\ab)}$ is large enough, for large enough $\ab$). This implies that an expected $\frac{9}{10}\ns$ of the $v_i$'s will be equal to $1$: more precisely, $\sum_{i=1}^\ns v_i \sim \binomial{\ns}{\tau}$ with $\tau \geq 9/10$. The probability that such a Binomial is less that $\ns/2$ is at most $1/10$ (and actually $e^{-\Omega(\ns)}$), and so the center will reject with probability at least $9/10$ (and actually $1-e^{-\Omega(\ns)}$).
\end{itemize}
\end{proof}
\subsection{Combined Algorithm for all values of \texorpdfstring{$\ms$}{m}}
\label{ssec:combining:nonprivate}

There exists a critical regime of parameters that~\cref{alg:hadamard-main} struggles with, leading to suboptimal sample complexity; and these are handled exactly  by~\cref{alg:large:m:tester}. We can handle both regimes as follows: have the server run both protocols, the former with $\ns_1=\ns-1$ and the latter with $\ns_2=1$ users, and return \accept if, and only if, both of them return \accept. Assume

\[
\ns \gtrsim \frac{\ab^{3/2} \log\ab}{\ms \dst^2} \lor \ab\,.
\]
Then, we have the following distinction of cases:
\begin{itemize}
    \item If $\p=\uniform_{\ab}$, then both protocols \accept with probability at least $9/10$ each, so overall the centre returns \accept with probability at least $8/10$;
    \item If $\totalvardist{\p}{\uniform_{\ab}} > \dst$, then
    \begin{itemize}
        \item If there exists $j^\ast \in [\ab]$ such that
        \[
            \dst_{j^\ast} = \bigOmega{\sqrt{\log(\ab)/\ms}}
        \]
        then, by~\cref{lemma:complementary}, the second protocol (with one user) outputs \reject with probability at least $9/10$, in which case the center outputs \reject.
        \item Otherwise, then by~\eqref{eq:lb:norm:mean:product} the mean of the product distribution in the first protocol is at least
        \begin{equation}
        \label{eq:cases:lb:mean:product}
            \normtwo{\mu}^2 
        \gtrsim \sum_{j=1}^{\ab} \Paren{\ms\dst_{j}^2\land 1} 
        \geq \sum_{j=1}^{\ab} \Paren{\frac{\ms\dst_{j}^2}{\log\ab}\land 1} 
        \gtrsim \frac{\ms\dst^2}{\log\ab}
        \end{equation}
        since $\ms\dst_{j}^2 \lesssim \log\ab$ for all $j$, and $\sum_{j=1}^{\ab} \dst_j^2 \geq \dst^2$. Then, concluding as in~\cref{sec:noname:hadamard:protocol} by invoking the uniformity testing algorithm for product distributions of~\cref{thm:unifprodtest} (which also handles the full range of distance parameter $\gamma^2 \eqdef \frac{\ms\dst^2}{\log\ab}  \in (0,\ab]$),
        the server rejects with probability at least $9/10$, as
        \[
            \frac{\ns}{\ab} \gtrsim \frac{\ab^{1/2}}{\gamma^2}\lor 1\,.
        \]
    \end{itemize}
    Either way, at least one of the two tests outputs \reject with probability $9/10$, and so does the center.
\end{itemize}

\subsection{Symmetric Protocols via Generalized Product Testing}
\label{ssec:symmetric:nonprivate}
The above protocol is asymmetric in that, as stated, it needs users to be divided into $\ab$ groups. Of these, $\ab-1$ groups will run~\cref{alg:hadamard-main}, each with a different column. The last group (of only 1 user) runs~\cref{alg:large:m:tester}. 

To resolve the asymmetry in~\cref{alg:hadamard-main} we prove the following statement, which covers the setting we require to make our protocol symmetric: each of $\ns' \eqdef \ns\dims$ users independently selects a random coordinate and reports a sample from that coordinate. Formally, observations $(X_1, j_1)\dots, (X_{\ns'},j_{\ns'})$ are obtained by choosing, independently for each $i\in[\ns']$, $j_i\in[\dims]$ uniformly at random, and $X_i \sim \p_{j_i}$. In this case, the numbers of times  $\ns_1,\dots,\ns_\dims$ each coordinate $j\in[\dims]$ is sampled are (correlated) $\binomial{\ns\dims}{1/\dims}$ random variables. 

\begin{theorem}
\label{thm:unifprodtest:symmetric}
There exists an algorithm (\cref{alg:mean:testing:symmetric}) which, given parameters $\gamma\in(0,\sqrt{\dims}]$, $\ns \geq 1$, and sample access to distributions $\p_1,\dots,\p_\dims$ on $\{-1,1\}$, chooses $\ns_1,\dots, \ns_\dims$ at random from a multinomial distribution with parameters $\ns'\eqdef \ns\dims$, $\dims$, and $(1/\dims,\dots,1/\dims)$; and then is given $\ns_j$ \iid samples from each $\p_j$ (where the samples are independent from the choice of $\ns_i$'s). Then, letting $\mu\eqdef\shortexpect_{X\sim \p_1\otimes\cdots\otimes\p_\dims}[X]\in[-1,1]^d$, it has the following guarantees.
\begin{itemize}
    \item If $\normtwo{\mu}\leq \frac{\gamma}{2}$, the algorithm returns $\accept$ w.p. $\geq \frac{2}{3}$;
    \item If $\normtwo{\mu}\geq \gamma$, the algorithm returns $\reject$ w.p. $\geq \frac{2}{3}$;
\end{itemize}
as long as $\ns\geq C\frac{\sqrt{\dims}}{\gamma^2}\lor 1$ for some absolute constant $C>0$.
\end{theorem}
\noindent We defer the proof to~\cref{proof:thm:unifprodtest:symmetric}.

As~\cref{thm:unifprodtest:symmetric} takes exactly the same parameters, and has exactly the same guarantees as~\cref{thm:unifprodtest}, we do not restate the proof of~\cref{thm:prcoin:1bit:assym}. We need only note that having each user randomly sample their own index $j$ and send it to the center is indeed distributed as described above. This of course increases communication to $\bigO{\log\ab}$ bits.

This alone does not make the entire protocol symmetric. As we said, we still have one user assigned to running~\cref{alg:large:m:tester}. However, this is easily addressed. At the cost of one extra bit of communication per user, we can simply have \emph{all} users run this second test, and then let the center select arbitrarily one of the $\ns$ outcomes to use. This finally gives a (non-private) algorithm, and can be summarized as follows:

\begin{theorem}[Symmetric private-coin uniformity testing]
    \label{thm:private:symmetric:combined:testing}

    There exists a \emph{symmetric, private-coin} (non-private) algorithm for uniformity testing (over domain size $\ab$) which on distance parameter $\dst\in(0,1]$ takes
    \[
    \ns = \bigO{\frac{\ab^{3/2}\log\ab}{\ms\dst^2}\lor \ab}
    \]
    users, each holding $\ms$ \iid samples from the unknown distribution, and sending $O(\log\ab)$ bits of communication.
\end{theorem}
Note that the first term dominates when $\ms \leq \bigO{\sqrt{\ab}/\dst^2}$, which is the regime of interest (as otherwise a single user has enough samples to run a uniformity testing algorithm by themselves).
\medskip

Of course, this result may still be somewhat underwhelming, as (albeit communication-efficient) the algorithm does not provide any privacy guarantee. In the next section, we will see how it can be easily adapted to yield our main result,~\cref{thm:private:combined:testing}.

\section{Symmetric Private Testing}
\label{sec:privacy}
We here analyze the private analogues to the algorithms defined in~\cref{sec:noname:hadamard:protocol}. We focus on making each of our two algorithms private before showing how they can be combined. In each case we will apply binary randomized response~\cite{Kasiviswanathan11} to the bit returned by each user.

Combining the two private algorithms gives us an asymmetric and symmetric protocol, each with sample complexity comparable to the $\ms=1$ case with $\ms$ times as many users, up to logarithmic factors. 

\begin{theorem}[1-bit user-level LDP uniformity testing]
    \label{thm:private:combined:testing}
    There exists a \emph{private-coin, user-level} locally differentially private algorithm for uniformity testing (over domain size $\ab$) which for small privacy parameter $\priv>0$, and distance parameter $\dst\in(0,1]$ takes
    \[
    \ns = \bigO{\frac{\ab^{3/2}\log(\ab/r)}{\ms\dst^2\priv^2}}
    \]
    users, each holding $\ms$ \iid samples from the unknown distribution. If the protocol is \emph{asymmetric}, $r=\priv$ and each user sends 1 bit of communication. If the protocol is \emph{symmetric}, then $r=\priv\dst\ms$ and each user sends $\bigO{\log\ab}$ bits of communication.
\end{theorem}

\subsection{Private~\cref{alg:hadamard-main} for small \texorpdfstring{$\ms$}{m}}

We first introduce the private version of~\cref{alg:hadamard-main}, which handles the case when $\ms$ is small. As described above we make this algorithm private by an application of binary randomized response to the bit returned by each user.

\begin{lemma}[Symmetric private-coin user-level LDP 1-bit uniformity tester]
\label{thm:prcoin:ldp:1bit}
    Given $\ns$ users, each holding $\ms$ samples of some unknown distribution $\p$ on $\ab$ elements. For small privacy parameter $\priv>0$. There exists an algorithm that distinguishes between $\p=\uniform_{\ab}$ and $\dtv(\p,\uniform_{\ab})>\dst$ with
    \[
    \ns = \bigO{\frac{\ab^{3/2}}{\priv^2(\ms\dst^2\wedge 1)}},
    \]
    samples. 
\end{lemma}
\begin{proof}

Let $Q\colon\{0,1\}\to\{0,1\}$ be binary randomized response as described in~\cref{sec:prelims}. We use the following simple fact about the distribution induced by $Q$: Let $X$ be a Bernoulli random variable with mean $p$, then the random variable obtained by applying $Q$ to $X$ has mean
\begin{align*}
    \bE{x\sim X}{Q(x)} &= \frac{e^\priv}{e^\priv + 1}p + (1-p)\frac{1}{e^\priv + 1}\\
    &= \frac{e^\priv - 1}{e^\priv + 1}p + \frac{1}{e^\priv + 1},
\end{align*}
and letting $p= \frac{1}{2}+\beta$, we get that  
\begin{align}
\label{eq:bias:private:coin}
    \bE{x\sim X}{Q(x)} &= \frac{1}{2} + \frac{e^\priv - 1}{e^\priv + 1}\beta .
\end{align}

Let $\tilde{Y}_i\defeq Q(Y_i)$ where $Y_i$ is the same bit returned by users in~\cref{alg:hadamard-main}. In order to apply~\cref{thm:unifprodtest:symmetric}, we do not use the deterministic assignment of groups as stated in~\cref{alg:hadamard-main}. Instead each user selects $j$ uniformly at random, and send $(\tilde{Y}_i,j)$ to the server. For notational convenience, we still allow $j(i)$ to denote the (randomized) selection of $j$ by user $i$. 

As before, the server computes 
\[
Z_i = 2(\tilde{Y}_i - 1/2) \in \{-1,1\} \tag{$1\leq i\leq \ns$}.
\]
By~\eqref{eq:bias:private:coin} and~\cref{lemma:epluribusunum}, $Z_i$ is therefore distributed as a Rademacher with mean 
\[
\abs{\beta_{j(i)}}=\bigOmega{2 \frac{e^\priv - 1}{e^\priv + 1}(\sqrt{\ms}\dst_{j(i)}\wedge 1)}.
\]

For brevity, let $\q_j$ denote the Rademacher distribution with mean $\pm\beta_{j}$. We therefore have the vector mean

\[
    \mu\eqdef\shortexpect_{X\sim \q_1\otimes\cdots\otimes \q_{\ab}}[X]=\left(\pm \beta_1,\dots, \pm \beta_{\ab} \right),
\]
Which, by~\cref{lemma:hadamard:norm:preservation}, has the appropriate norm
\begin{equation}
\label{eq:lb:private:norm:mean}
\normtwo{\mu}^2 
    = \sum_{j=1}^{\ab} \beta_j^2 
    \gtrsim 4\Paren{\frac{e^\priv - 1}{e^\priv + 1}}^2 \sum_{j=1}^{\ab} (\ms\dst_j^2 \wedge 1)\geq 4\Paren{\frac{e^\priv - 1}{e^\priv + 1}}^2 (\ms\dst^2 \wedge 1).
\end{equation}
Plugging into~\cref{thm:unifprodtest:symmetric} with $\ns'=\ns(\ab-1)$, gives a final sample complexity of
\begin{align*}
    \frac{\ns}{\ab-1}\gtrsim \frac{\sqrt{\ab}(e^\priv + 1)^2}{4(e^\priv - 1)^2 (\dst^2\ms\wedge 1)}
\end{align*}
implying that for small $\priv$, we have 
\[
\ns=\bigO{\frac{\ab^{3/2}}{\priv^2(\ms\dst^2\wedge 1)}}.
\]
Noting that as $\priv$ grows, this rapidly converges to the non-private result.
    
\end{proof}

\subsection{Private~\cref{alg:large:m:tester} for large \texorpdfstring{$\ms$}{m}}
Recall that the second stage of our algorithm, defined in~\cref{sec:complementary:cornercase} handles the case when one of the subsets defined by the Hadamard matrix is overrepresented. Non-privately we only require that a single user perform this test for all of the subsets and report their response. Under local differential privacy this would not have any high-probability guarantee. Instead we use a known (and easily derived) bound for learning coins under binary randomized response. Specifically, that one can learn a Bernoulli through randomized response up to additive error $\alpha$ with success probability 2/3 using $\ns=\bigO{1/(\alpha^2\priv^2)}$ samples.

\begin{lemma}[Private version of~\cref{lemma:complementary}]
    \label{lemma:private:complementary}
    \Cref{alg:large:m:tester}, when each user applies binary randomized response to their output, has the following guarantees:
    \begin{itemize}
        \item If $\p=\uniform_{\ab}$, then the server outputs \accept with probability at least $\frac{2e^{\priv} + 1}{3(e^\priv + 1)}$
        \item If there exists $1\leq j^* \leq \ab$ such that $\dst_{j^*}=\bigOmega{\sqrt{\log(\ns\ab)/\ms}}$, then the server outputs \reject with probability at least $\frac{2e^{\priv} + 1}{3(e^\priv + 1)}$
    \end{itemize}
\end{lemma}
\noindent This follows directly from applying randomized response to each users' output from~\cref{alg:large:m:tester} and applying the fact stated above to learn the coin up to additive error $1/9$.

\subsection{Analysis of the Combined Algorithm}
As in the non-private case, we have to consider how to combine these two algorithms. First, we consider the ``easy'' asymmetric case; have $\ns_1$ users run the algorithm described in the proof of~\cref{thm:prcoin:ldp:1bit}, we then only need $\ns_2=\bigO{1/\priv^2}$ users to run the second protocol. As the number of users required by the second protocol is clearly dominated by the first, we retain much the same sample complexity (up to a logarithmic factor described below). 

To make the protocol symmetric we can follow the same procedure described in~\cref{ssec:symmetric:nonprivate} and have \emph{all} users run the second protocol. This incurs two costs, (1) we must divide the privacy budget between these two protocols, and (2) we lose a logarithmic $\ns$ factor in the final sample complexity.

We here describe both approaches and their final sample complexities.

\begin{proofof}{{\cref{thm:private:combined:testing}}}
We begin by assuming that $\ns_1$ users run the private version of~\cref{alg:hadamard-main}, and $\ns_2$ users run the private version of~\cref{alg:large:m:tester}. Once again, the server outputs \accept if both tests return \accept, and otherwise it outputs \reject. We defer, for now, the decision of how best to choose $\ns_1$ and $\ns_2$. Assume
\[
\ns\gtrsim \frac{\ab^{3/2}\log(\ns_2 \ab)(e^\priv + 1)^2}{\ms\dst^2(e^\priv - 1)^2}\vee \ab,
\]
noting that the first term dominates in our regime of interest, $\ms<\sqrt{\ab}/\dst^2$. 

First, consider $\p=\uniform_{\ab}$. By~\cref{lemma:private:complementary,thm:prcoin:ldp:1bit}, the first protocol accepts with probability at least $2/3$ and the second with probability at least $\frac{2e^\priv + 1}{3(e^\priv + 1)}$, and so the probability that both accept is at least $8/9$. This resolves the uniform case.

Now, when $\dtv(\p,\uniform_{\ab})>\dst$ we again have a distinction of cases. If there exists some $j^{*}\in[\ab]$ such that $\dst_{j^*}=\bigOmega{\sqrt{\log(\ns_2 \ab)/\ms}}$, then by~\cref{lemma:private:complementary}, the second protocol outputs \reject with probability at least $\frac{2e^\priv + 1}{3(e^\priv + 1)}$ and so the server outputs \reject.

Otherwise, we have that by~\eqref{eq:lb:private:norm:mean}, the mean if the product distribution in the first protocol is at last

\begin{equation}
\label{eq:cases:private:mean:product}
    \normtwo{\mu}^2 
    \gtrsim \Paren{\frac{e^\priv - 1}{e^\priv + 1}}^2\sum_{j=1}^{\ab} \Paren{\ms\dst_{j}^2\land 1} 
    \geq \Paren{\frac{e^\priv - 1}{e^\priv + 1}}^2\sum_{j=1}^{\ab} \Paren{\frac{\ms\dst_{j}^2}{\log(\ns_2 \ab)}\land 1} 
    \gtrsim \frac{\ms\dst^2(e^\priv - 1)^2}{\log(\ns_2 \ab)(e^\priv + 1)^2}
\end{equation}
since $\ms\dst_{j}^{2}\lesssim \log({\ns_2 \ab})$ for all $j$ and $\sum_{j=1}^{\ab} \dst_{j}^2 > {\dst}^{2}$. We then invoke the uniformity testing algorithm for product distributions of~\cref{thm:unifprodtest:symmetric} with $\gamma^{2}\defeq \frac{(e^\priv - 1)^2 \ms\dst^{2}}{(e^\priv + 1)^2 \log(\ns_2 \ab)}$ and $\ns' = \ns_1\ab$. The server rejects with probability at least 2/3, as
\[
\frac{\ns_1}{\ab}\gtrsim \frac{\ab^{1/2}}{\gamma^2}.
\]

Therefore, in the uniform case, each test accepts with high probability, and so does the server. And in the far-from-uniform case, at least one of the two tests outputs \reject with probability 2/3, and so the server does also.

We close by selecting the appropriate value for $\ns_2$ in the symmetric and asymmetric cases. In the asymmetric case we assign $\ns_2=\bigO{1/\priv^2}$ to run~\cref{alg:large:m:tester}, while leaving the rest to run~\cref{alg:hadamard-main}. The number of users $\ns_1$ clearly dominates $\ns_2$ in all regimes, and so we get a final sample complexity of
\[
\ns\gtrsim \frac{\ab^{3/2}}{\ms\dst^2(\priv^2 \lor 1)}\log\Paren{\frac{\ab}{\priv\lor 1}}.
\]

In the symmetric case, things are complicated, especially with regards to practicality. In order to use our established strategy of ``have everyone run both algorithms'' we must divide our privacy budget between the two algorithms. We here take the path of least resistance, and assume that each user runs the private version of~\cref{alg:hadamard-main} with privacy parameter $\priv_1=\priv/2$, and likewise runs~\cref{alg:large:m:tester} with $\priv_2 = \priv/2$. As such, we have $\ns_2=\ns_1=\ns$ and so gain an $\ns_2$ term in the log of~\cref{lemma:private:complementary}. Setting $\ns_2$ to the sample complexity derived in~\cref{thm:prcoin:ldp:1bit}, one can see that we require
\[
\ns\gtrsim \frac{\ab^{3/2}}{\ms\dst^2(\priv^2 \lor 1)}\log\Paren{\frac{\ab}{\ms\dst(\priv\lor 1)}}.
\]
users.    
\end{proofof}

Naturally, one may consider some obvious improvements to the last stages of the proof. We state two such possible improvements here, but do not go further into their details as they affect only constant or logarithmic terms in our final bound.

The first of such improvements one may consider would be to divide the privacy budget more economically. By setting $\priv_1$ proportional to the sample complexity of~\cref{alg:hadamard-main}, and $\priv_2$ likewise, one could almost-surely save some very real constant factors, and would have the benefit of seeing the algorithms converge to their answer at similar rates. 

The other, perhaps less obvious (and more cumbersome to analyze), approach would be to have each user flip a coin to decide which protocol they will run. They would then send their decision to the server along with their response. This would seem to reduce the logarithmic term in the sample complexity of the symmetric protocol down to that of the asymmetric protocol.

\section{Conclusion}
\label{sec:conclusion}
User-level locally private distribution testing is still far from being understood. We observe many phase transitions as $\priv$  and $\ms$ vary. Consider the algorithm for testing in the central model of differential privacy discussed in~\cref{ssec:related}, when $\ms$ exceeds the stated sample complexity we see that the required number of ``users'' $\ns$ goes to 1, and only 1 bit of communication is needed.

How exactly do these algorithms behave as a sliding scale between the local and central models of differential privacy? Characterizing the behaviour in each regime is an ongoing and important field of research.

\paragraph{Future work.} 
Data generation is not always homogenous: the distributions that users sample from are not truly identical; rather, it is more likely they are sampling from $\ns$ distributions that could be similar or very far apart. User-level locally private distribution learning under limited heterogeneity is touched upon in~\cite{AcharyaLS23}, but we mirror their remark that this deserves further study.

A further heterogeneity that should be considered is the case when not all users hold the same number of samples $\ms$. In this case each may hold $\ms_i$ for each $i\in[\ns]$. It is not at all obvious how this could be handled neatly, and general results for this model could greatly help practical implementations.

\printbibliography
\appendix

\section{Deferred Proofs}
\label{app:deferred}
\subsection{Mean Testing for Product Distributions}
    \label{sec:product:testing}
    In this appendix, we first provide, for completeness, a self-contained proof of~\cref{thm:unifprodtest} (restated below), before establishing \cref{thm:unifprodtest:symmetric}, our new ``symmetric mean testing protocol'' which we relied on to make our main protocol symmetric.
\begin{lemma}[\cref{thm:unifprodtest}, restated]
\label{thm:unifprodtest:restated}
There exists an algorithm which, given a parameter $\gamma\in(0,\sqrt{\dims}]$ and $\ns$ \iid samples from a product distribution $\p$ on $\{-1,1\}^\dims$ with $\mu\eqdef\shortexpect_{X\sim \p}[X]\in[-1,1]^d$, has the following guarantees.
\begin{itemize}
    \item If $\normtwo{\mu}\leq \frac{\gamma}{2}$, the algorithm returns $\accept$ w.p. $\geq \frac{2}{3}$;
    \item If $\normtwo{\mu}\geq \gamma$, the algorithm returns $\reject$ w.p. $\geq \frac{2}{3}$;
\end{itemize}
as long as $\ns\geq C\frac{\sqrt{\dims}}{\gamma^2}$ for some absolute constant $C>0$. (Moreover, one can take $C=50$.)
\end{lemma}
\begin{proof}
    Hereafter, we assume $\ns \geq 2$ (which will follow from choosing $C\geq 2$). The algorithm is quite simple: it first computes, given the $\ns$ \iid samples $X^{(1)},\dots,X^{(\ns)}$ from $\p$, the empirical mean
    \[
        \bar{X} \eqdef \frac{1}{\ns}\sum_{i=1}^{\ns} X^{(i)} \in [-1,1]^\dims
    \]
    and outputs \accept if, and only if, $\normtwo{\bar{X}}^2 \leq \frac{\dims}{\ns}+ \frac{\gamma^2}{2}$. 

    To show this succeeds with high probability as claimed, we analyze the expectation and variance of $Z\eqdef \normtwo{\bar{X}}^2$, and conclude by Chebychev's inequality. To simplify the computations, we record first the fact that, for every $j\in[\dims]$, the random variable 
    \[
    N_j \eqdef \ns\frac{\bar{X}_j+1}{2} = \sum_{i=1}^{\ns} \frac{X^{(i)}_j+1}{2}
    \]
    is Binomially distributed: $N_j\sim\binomial{\ns}{\frac{\mu_j+1}{2}}$; moreover, the $N_j$'s are mutually independent.
    \begin{claim}
        \label{claim:meantesting:simple:expectation}
    We have
    \[
        \expect{Z} = \frac{\dims}{\ns}+\frac{\ns-1}{\ns}\normtwo{\mu}^2\,,
    \]
    \end{claim}
    \begin{proof}
        By linearity of expectation, and since $\ns\bar{X}_j = 2N_j - \ns$ for all $j\in[\dims$,
        \begin{align*}
            \expect{Z}
            &= \sum_{j=1}^{\dims} \expect{\bar{X}_j^2}
            = \frac{1}{\ns^2}\sum_{j=1}^{\dims} \expect{(\ns\bar{X}_j)^2}\\
            &= \frac{4}{\ns^2}\sum_{j=1}^{\dims} \expect{\Paren{N_j - \frac{\ns}{2}}^2} \\
            &= \frac{4}{\ns^2}\sum_{j=1}^{\dims} \expect{\Paren{(N_j - \expect{N_j}) + \Paren{\expect{N_j} - \frac{\ns}{2}}}^2} \\
            &= \frac{4}{\ns^2}\sum_{j=1}^{\dims} \Paren{ \ns \frac{1-\mu_j^2}{4} + \ns^2\frac{\mu_j^2}{4} } \tag{Variance and expectation of $N_j$}\\
            &= \frac{\dims}{\ns}+\frac{\ns-1}{\ns}\normtwo{\mu}^2\,,
        \end{align*}
        as claimed.
    \end{proof}
    Turning to the variance, we get the following:
    \begin{claim}
        \label{claim:meantesting:simple:variance}
    We have
    \[
        \var[Z]
            \leq \frac{2\dims}{\ns^2}+\frac{4(\ns-1)}{\ns^2}\cdot \normtwo{\mu}^2\,.
    \]
    \end{claim}
    \begin{proof}
        By independence of the $\bar{X}_i$'s, and proceeding as in the previous claim we get
        \begin{align*}
            \var[Z]
            &= \sum_{j=1}^{\dims} \var[\bar{X}_j^2]
            = \frac{1}{\ns^4}\sum_{j=1}^{\dims} \var[(\ns\bar{X}_j)^2]\\
            &= \frac{16}{\ns^4}\sum_{j=1}^{\dims} \var\mleft[\Paren{N_j - \frac{\ns}{2}}^2\mright] 
        \end{align*}
        For any $j\in[\dims]$, we have $\var\mleft[\Paren{N_j - \frac{\ns}{2}}^2\mright]
        = \expect{\Paren{N_j - \frac{\ns}{2}}^4} - \expect{\Paren{N_j - \frac{\ns}{2}}^2}^2$; and we already computed the value of the second term in the previous claim, as
        \begin{equation}
        \expect{\Paren{N_j - \frac{\ns}{2}}^2}^2
        = \Paren{ \ns \frac{1-\mu_j^2}{4} + \ns^2\frac{\mu_j^2}{4} }^2
        = \frac{\ns^2}{16}\Paren{ 1+(\ns-1)\mu_j^2}^2
        = \frac{\ns^2}{16}\Paren{ 1+2(\ns-1)\mu_j^2+(\ns-1)^2\mu_j^4}\,. \label{eq:meantesting:simple:variance:1}
        \end{equation}
        The other one can be computed as follows (where we eventually rely on the known expression for the central moments of a Binomial distribution, and the fact that $\expect{N_j}- \frac{\ns}{2} = \frac{1}{2}\ns\mu_j$):
        \begin{align}
            \expect{\Paren{N_j - \frac{\ns}{2}}^4}
            &= \expect{\Paren{\Paren{N_j - \expect{N_j}} + \Paren{\expect{N_j}- \frac{\ns}{2}}}^4} \notag\\
            &= \expect{\Paren{N_j - \expect{N_j}}^4} \notag\\
            &\qquad +4\expect{\Paren{N_j - \expect{N_j}}^3}\Paren{\expect{N_j}- \frac{\ns}{2}}\notag\\
            &\qquad +6\expect{\Paren{N_j - \expect{N_j}}^2}\Paren{\expect{N_j}- \frac{\ns}{2}}^2\notag\\
            &\qquad +4\expect{\Paren{N_j - \expect{N_j}}}\Paren{\expect{N_j}- \frac{\ns}{2}}^3\notag\\
            &\qquad +\Paren{\expect{N_j}- \frac{\ns}{2}}^4\notag\\
            &= \expect{\Paren{N_j - \expect{N_j}}^4} 
            + 2\ns\mu_j\expect{\Paren{N_j - \expect{N_j}}^3}
            + \frac{3}{2}\ns^2\mu_j^2\expect{\Paren{N_j - \expect{N_j}}^2}
            + \frac{\ns^4\mu_j^4}{16}\notag\\
            &= \ns\frac{1-\mu_j^2}{4}\Paren{1+3(\ns-2)\frac{1-\mu_j^2}{4}}
            - \ns^2\mu_j^2\frac{1-\mu_j^2}{2}
            + \frac{3}{8}\ns^3\mu_j^2(1-\mu_j^2)
            + \frac{\ns^4\mu_j^4}{16} \label{eq:meantesting:simple:variance:2}
        \end{align}
        Subtracting~\eqref{eq:meantesting:simple:variance:1} from~\eqref{eq:meantesting:simple:variance:2} and keeping track of all the terms, we get
        \begin{equation}
        \var\mleft[\Paren{N_j - \frac{\ns}{2}}^2\mright] 
        = \frac{\ns(\ns-1)}{8}
        + \frac{\ns(\ns-1)(\ns-2)}{4}\cdot \mu_j^2
        - \frac{\ns(\ns-1)(2\ns-3)}{8}\cdot \mu_j^4\,,
        \end{equation}
        which, dropping the last (non-positive, as $\ns \geq 2$) term and summing over $j$, yields
        \begin{align*}
            \var[Z]
            &= \frac{16}{\ns^4}\sum_{i=1}^{\dims} \var\mleft[\Paren{N_j - \frac{\ns}{2}}^2\mright] 
            \leq \frac{2\dims}{\ns^2}+\frac{4(\ns-1)}{\ns^2}\cdot \normtwo{\mu}^2\,,
        \end{align*}
        as sought.
    \end{proof}
    By the above two claims, we have
    \[
    \var[Z] \leq \frac{2\dims}{\ns^2}+ \frac{4}{\ns} \Paren{ \expect{Z} - \frac{\dims}{\ns}}\,.
    \]
    We can now conclude by Chebyshev's inequality:
    \begin{itemize}
        \item If $\normtwo{\mu} \leq \frac{\gamma}{2}$, then, by~\cref{claim:meantesting:simple:expectation},
        $\expect{Z} \leq \frac{\dims}{\ns} + \frac{\gamma^2}{4}$, and so
        \[
        \probaOf{Z > \frac{\dims}{\ns} + \frac{\gamma^2}{2}} \leq \frac{\var[Z]}{(\gamma^2/4)^2}
        \leq \frac{16}{\gamma^4}\Paren{\frac{2\dims}{\ns^2}+ \frac{4}{\ns} \cdot \frac{\gamma^2}{4} }
        = \frac{32\dims}{\ns^2\gamma^4}+ \frac{16}{\ns\gamma^2} 
        \]
       which is at most $1/3$ for $\ns \geq 50\frac{\sqrt{\dims}}{\gamma^2}$.
        \item If $\normtwo{\mu} \geq \gamma$, then, by~\cref{claim:meantesting:simple:expectation},
        $\expect{Z} \geq \frac{\dims}{\ns} + \gamma^2$; letting $\Delta \eqdef \expect{Z} - \frac{\dims}{\ns} \geq \gamma^2$, we have
        \[
        \probaOf{Z \leq \frac{\dims}{\ns} + \frac{\gamma^2}{2}} \leq
        \probaOf{Z \leq \frac{\dims}{\ns} + \frac{\Delta}{2}} \leq \frac{\var[Z]}{(\Delta/2)^2}
        \leq \frac{4}{\Delta^2}\Paren{\frac{2\dims}{\ns^2}+ \frac{4}{\ns} \cdot \Delta }
        \leq \frac{8\dims}{\ns^2\gamma^4}+ \frac{16}{\ns\gamma^2}
        \]
        which is also at most $1/3$ when $\ns \geq 50\frac{\sqrt{\dims}}{\gamma^2}$.
    \end{itemize}
    This concludes the proof.
\end{proof}

We now turn to proving~\cref{thm:unifprodtest:symmetric}, our generalization of~\cref{thm:unifprodtest}, restated below.
\begin{theorem}[\cref{thm:unifprodtest:symmetric}, restated]
\label{thm:unifprodtest:symmetric:restated}
There exists an algorithm (\cref{alg:mean:testing:symmetric}) which, given parameters $\gamma\in(0,\sqrt{\dims}]$, $\ns \geq 1$, and sample access to distributions $\p_1,\dots,\p_\dims$ on $\{-1,1\}$, chooses $\ns_1,\dots, \ns_\dims$ at random from a multinomial distribution with parameters $\ns'\eqdef \ns\dims$, $\dims$, and $(1/\dims,\dots,1/\dims)$; and then is given $\ns_j$ \iid samples from each $\p_j$ (where the samples are independent from the choice of $\ns_i$'s). Then, letting $\mu\eqdef\shortexpect_{X\sim \p_1\otimes\cdots\otimes\p_\dims}[X]\in[-1,1]^d$, it has the following guarantees.
\begin{itemize}
    \item If $\normtwo{\mu}\leq \frac{\gamma}{2}$, the algorithm returns $\accept$ w.p. $\geq \frac{2}{3}$;
    \item If $\normtwo{\mu}\geq \gamma$, the algorithm returns $\reject$ w.p. $\geq \frac{2}{3}$;
\end{itemize}
as long as $\ns\geq C\frac{\sqrt{\dims}}{\gamma^2}\lor 1$ for some absolute constant $C>0$.
\end{theorem}
\begin{algorithm}[htbp]
    \begin{algorithmic}
        \Require Parameters $\ns \geq 2, \dims \geq 2, \gamma \in(0,\sqrt{\dims}]$, sample access to $\p_1,\dots, \p_\dims$ over $\{-1,1\}$
        \Ensure Returns \yes if the mean vector of $\p_1\otimes\cdots\otimes\p_\dims$ has $\lp[2]$ norm at most $\gamma/2$, and \no if it is at least $\gamma$
        \State Draw $(\ns_1,\dots, \ns_\dims)\sim \operatorname{Multinom}(\ns\dims, \dims, \frac{1}{\dims}\mathbf{1}_\dims)$
        \ForAll{$1\leq j\leq \dims$}
            \State Get $\ns_j$ \iid samples $X^{(1)}_j,\dots,X^{(\ns_j)}_j$ from $\p_{j}$
            \State Compute
            \[
                \bar{X}_j \gets \frac{1}{\ns}\sum_{i=1}^{\ns_j} X^{(i)}_j
            \]
        \EndFor
        \State Compute $Z \gets \normtwo{\bar{X}}^2 - \frac{\dims}{\ns}$, where 
        $\bar{X} = (\bar{X}_1,\dots,\bar{X}_\dims)$
        \If{$Z \leq \frac{\gamma^2}{4}$}
            \State \Return \yes \Comment{Small $\lp[2]$ norm}
        \EndIf
        \State \Return \no \Comment{Large $\lp[2]$ norm}
    \end{algorithmic}
    \caption{Bandit mean testing algorithm for Rademacher product distributions}
    \label{alg:mean:testing:symmetric}
\end{algorithm}
\begin{proofof}{{\cref{thm:unifprodtest:symmetric}}}
\label{proof:thm:unifprodtest:symmetric}
First, observe that in this setting, the random variables $\ns_1,\dots,\ns_\dims$  are known to satisfy the property of \emph{negatively association} (see, \eg~\cite{DubhashiR98,Wajc17}), a stronger condition than negative correlation which will need later in the proof (specifically, for~
\cref{claim:meantesting:symmetric:variance}). Moreover, they are marginally distributed as $\binomial{\ns\dims}{1/\dims}$ r.v.'s, with common expectation 
$
\ns
$
and variance
$
\ns\Paren{1-\frac{1}{\dims}}
$. Moreover, their (negative) covariances can be explicitly computed:
for any $1\leq j < j' \leq \dims$,
$
\cov(\ns_j,\ns_{j'})
= -\frac{\ns}{\dims} \leq 0
$.\medskip

 Let $\bar{X}_1,\dots, \bar{X}_\dims$ be defined as
    \[
        \bar{X}_j \eqdef \frac{1}{\ns}\sum_{i=1}^{\ns_j} X^{(i)}_j,\qquad 1\leq j\leq \dims\,,
    \]
    where $X^{(1)}_j,\dots, X^{(\ns_j)}_j$ are the $\ns_j$ \iid samples from $\p_j$ (where $\ns_j$ itself is a random variable). Finally, 
    let $\bar{X} \eqdef (\bar{X}_1,\dots,\bar{X}_\dims)$. We then define the statistic
    \begin{equation}
        Z \eqdef \normtwo{\bar{X}}^2 - \frac{\dims}{\ns}\,.
    \end{equation}
    \begin{claim}
        \label{claim:meantesting:symmetric:expectation}
    We have
    \[
        \expect{Z} = \Paren{\frac{\ns-1}{2\ns}+\frac{\sigma^2}{2\ns^2}}\normtwo{\mu}^2\,,
    \]
        where the randomness is taken over the choice of $\ns_1,\dots, \ns_\dims$ and the observations $(X_j^{(i)})_{j\in[\dims],i\in[\ns_j]}$.
    \end{claim}
    \begin{proof}
    Fix any $j\in[\dims]$. By the Law of Total Expectation, and using the fact that $|X_j^{(i)}|=1$ for all $i,j$,
    \begin{align*}
    \ns^2\expect{\bar{X}_j^2}
    &= \ns^2\expect{\expectCond{\bar{X}_j^2}{\ns_j}}
    = \expect{\sum_{i=1}^{\ns_j}\sum_{i'=1}^{\ns_j} \expectCond{X^{(i)}_j X^{(i')}_j}{\ns_j}} \\
    &= \expect{\sum_{i=1}^{\ns_j} \expectCond{1}{\ns_j}
    + 2\sum_{1\leq i< i' \leq \ns_j} \expectCond{X^{(i)}_j}{\ns_j} \expectCond{X^{(i')}_j}{\ns_j} } \\
    &= \expect{\ns_j
    + \binom{\ns_j}{2} \mu_j^2 } 
    = \ns
    + \expect{\binom{\ns_j}{2}}\mu_j^2
    = \ns
    + \binom{\ns}{2}\mu_j^2+ \frac{\var[\ns_j]}{2}\mu_j^2\,.
    \end{align*}
    \vmargin{Because we're splitting the expectations in the above, it might make it clearer to add what the expectations are over. Not the most urgent thing to change.}
        From there,
        \begin{align*}
            \expect{\normtwo{\bar{X}}^2}
            &= \sum_{j=1}^\dims \expect{\bar{X}_j^2}
            = \frac{\dims}{\ns}
            + \frac{\ns-1}{2\ns}\normtwo{\mu}^2
            + \frac{\sigma^2}{2\ns^2}\normtwo{\mu}^2\,,
        \end{align*}
        establishing the claim.
    \end{proof}
Turning to the variance of $Z$, we have the following bound:
    \begin{claim}
        \label{claim:meantesting:symmetric:variance}
    We have
    \[
        \var[Z] \leq \frac{\dims}{\ns^2} + \frac{2\normtwo{\mu}^4}{\ns} + \frac{2\normtwo{\mu}^2  }{\ns}
    \]
        where the randomness is taken over the choice of $\ns_1,\dots, \ns_\dims$ and the observations $(X_j^{(i)})_{j\in[\dims],i\in[\ns_j]}$.
    \end{claim}
    \begin{proof}
        Since $\var[Z] = \var[\normtwo{\bar{X}}^2]$, we focus on the latter, writing $Y \eqdef \normtwo{\bar{X}}^2$ for convenience. 
        By the Law of Total Variance, letting $N \eqdef (\ns_1,\dots, \ns_\dims)$,
        \[
        \var[Y] = \expect{\var[Y\mid N]} + \var[\expectCond{Y}{N}]\,.
        \]
        We start by computing the second term. We already, in the proof of the previous claim, (implicitly) calculated $\expectCond{Y}{N}$, showing that
        \[
        \expectCond{Y}{N} = \frac{1}{\ns^2}\sum_{j=1}^\dims \Paren{\ns_j + \binom{\ns_j}{2}\mu_j^2}\,,
        \]
        Using negative association (since $x\mapsto x+\binom{x}{2}\mu_j^2$ is non-decreasing), one can show\cmargin{should show? Also, check.}\vmargin{Definitely need to show how we got this.} that
        \[
        \var[\expectCond{Y}{N}] \leq \frac{1}{\ns^4}\sum_{j=1}^\dims \var\Paren{\ns_j + \binom{\ns_j}{2}\mu_j^2}\,,
        \]
        from which a direct (and tedious) computation yields
        \begin{align}
             \var[\expectCond{Y}{N}] 
             &\leq \frac{1}{\dims^3\ns^4}\sum_{j=1}^\dims
             \Paren{\ns\dims^3 + 2\ns^2\dims^3\mu_j^2 + 2\ns^3\dims^3 \mu_j^4} \notag\\
             &= \frac{\dims}{\ns^3} + \frac{2}{\ns^2}\normtwo{\mu}^2+ \frac{2}{\ns}\norm{\mu}_4^4 \notag\\
             &\leq \frac{\dims}{\ns^3} + \frac{2}{\ns^2}\normtwo{\mu}^2+ \frac{2}{\ns}\normtwo{\mu}^4\,. \label{eq:firstcomponent:totalvariance} 
        \end{align}
        For the first term, we start by computing $\var[Y\mid N]$:
        \begin{align*}
            \var[Y\mid N]
            &= \var[ \sum_{j=1}^\dims \bar{X}_j^2 \mid N] \\
            &= \sum_{j=1}^\dims \var[ \bar{X}_j^2 \mid N] \tag{conditional independence} \\
            &= \frac{1}{\ns^4}\sum_{j=1}^\dims \Paren{\binom{\ns_j}{2}+ 6\binom{\ns_j}{3}\mu_j^2 + \Paren{6\binom{\ns_j}{4}-\binom{\ns_j}{2}^2}\mu_j^4} \tag{explicit computation} \\
            &\leq \frac{1}{\ns^4}\sum_{j=1}^\dims \Paren{\binom{\ns_j}{2} + 6\binom{\ns_j}{3}\mu_j^2}\,. 
        \end{align*}
        This leads to, since $\ns_j \sim\binomial{\ns\dims}{\frac{1}{\dims}}$,\footnote{If $X\sim\binomial{m}{p}$, then $\expect{\binom{X}{2}}= \binom{m}{2}p^2$ and $\expect{\binom{X}{3}}= \binom{m}{3}p^3$.}
        \begin{align}
            \expect{\var[Y\mid N]}
            &\leq \frac{1}{\ns^4}\sum_{j=1}^\dims \Paren{\expect{\binom{\ns_j}{2}} + 6\mu_j^2\expect{\binom{\ns_j}{3}}} \notag\\
            &= \frac{\ns(\ns\dims-1)}{2\ns^4} + 6\frac{\ns\dims(\ns\dims-1)(\ns\dims-2)}{6\dims^3\ns^4}\normtwo{\mu}^2  \notag\\
            &\leq \frac{\dims}{2\ns^2} + \frac{\normtwo{\mu}^2}{\ns}\,.  \label{eq:secondcomponent:totalvariance} 
        \end{align}
        Combining~\eqref{eq:firstcomponent:totalvariance} and~\eqref{eq:secondcomponent:totalvariance}, we get
        \begin{align*}
        \var[Z]
        = \var[Y]
        &\leq \frac{\dims}{\ns^3} + \frac{2\normtwo{\mu}^2}{\ns^2}+ \frac{2\normtwo{\mu}^4}{\ns} + \frac{\dims}{2\ns^2} + \frac{\normtwo{\mu}^2  }{\ns} \\
        &\leq \frac{2\normtwo{\mu}^4}{\ns} + \frac{\dims}{\ns^2} + \frac{2\normtwo{\mu}^2  }{\ns} \tag{as $\ns \geq 2$}
        \end{align*}
        as claimed.
    \end{proof}
    We are now able to conclude the proof of~\cref{thm:unifprodtest:symmetric}, invoking the expectation and variance analysis of~\cref{claim:meantesting:symmetric:expectation,claim:meantesting:symmetric:variance}. Recall that $\sigma^2 = \ns(1-1/\dims)$.
    \begin{itemize}
        \item If $\normtwo{\mu}\leq \frac{\gamma}{2}$, then $\expect{Z} \leq \frac{\normtwo{\mu}^2}{2} \leq \frac{\gamma^2}{8}$, and
        \begin{align*}
            \probaOf{Z > \frac{\gamma^2}{4} } 
            &\leq \frac{64\var[Z]}{\gamma^4} \tag{by Chebychev}\\
            &\leq \frac{64\dims}{\ns^2\gamma^4} + \frac{128\normtwo{\mu}^4}{\ns\gamma^4} + \frac{128\normtwo{\mu}^2  }{\ns\gamma^4} \tag{by \cref{claim:meantesting:symmetric:variance}}\\
            &\leq \frac{64\dims}{\ns^2\gamma^4} + \frac{8}{\ns} + \frac{32}{\ns\gamma^2} 
        \end{align*}
        which is at most $1/3$ for $\ns \geq C_1\cdot \max\Paren{\frac{\sqrt{\dims}}{\gamma^2}, 1}$ (for a sufficiently large absolute constant $C_1>0$).
        \item If $\normtwo{\mu}\geq \gamma$, then by~\cref{claim:meantesting:symmetric:expectation} we have $\expect{Z} \geq \frac{1}{2}\Paren{1-\frac{1}{\ns\dims}}\normtwo{\mu}^2 \geq \frac{3}{4}\gamma^2$ (using $\ns,\dims\geq 2$), and so
        \begin{align*}
            \probaOf{Z \leq \frac{\gamma^2}{4} }
            &\leq \probaOf{Z \leq \frac{1}{3}\expect{Z} } \\
            &\leq \probaOf{\abs{\expect{Z}-Z} \geq \frac{2}{3}\expect{Z} } \\
            &\leq \frac{9\var[Z]}{4\expect{Z}^2} \tag{by Chebychev}\\
            &\leq \frac{9\dims}{4\ns^2\expect{Z}^2} + \frac{9\normtwo{\mu}^4}{2\ns\expect{Z}^2} + \frac{9\normtwo{\mu}^2}{2\ns\expect{Z}^2} \tag{by \cref{claim:meantesting:symmetric:variance}}\\
            &\leq \frac{4\dims}{\ns^2\gamma^4} + \frac{8}{\ns} + \frac{8}{\ns\gamma^2}  \tag{as $\expect{Z} \geq \frac{3}{4}\normtwo{\mu}^2$}
        \end{align*}
        which is at most $1/3$ for $\ns \geq C_2\cdot \max\Paren{\frac{\sqrt{\dims}}{\gamma^2}, 1}$ (for a sufficiently large absolute constant $C_2>0$).
    \end{itemize}
    Choosing $C \eqdef \max(C_1,C_2)$ concludes the proof.
\end{proofof}

\subsection{The Hadamard transform is norm-preserving}
    \label{app:hadnormpreserved}
    The following proof is included for completeness, but follows much the same route as~\cite{ACFT:19}.

\begin{proofof}{\cref{lemma:hadamard:norm:preservation}}
\label{proof:hadamard:norm}
By the definition of $\p(\chi_j)$ we have that

\begin{align*}
\p(\chi_{j}) & =\frac{1}{2}(h_{j} + \mathbf{1}_{\ab})\p \\
 & = \frac{1}{2}(h_{j}\p + 1)
\end{align*}

where $h_{j}\in \{\pm1\}^{\ab}$ is the $j$'th column of the Hadamard matrix $H$, and $\mathbf{1}_{\ab}$ is the all-ones vector of order $\ab$. Defining $\p(\chi)=(\p(\chi_{1}),\dots,\p(\chi_{\ab}))$ we see that

\begin{align*}
    \p(\chi)=\frac{1}{2}(H\p + \mathbf{1}_{\ab}).
\end{align*}

\noindent Returning to the norm, we have
\begin{align*}
\lVert \p(\chi)-\uniform_{\ab}(\chi) \rVert_{2}^2  & = (\p(\chi)-\uniform_{\ab}(\chi))^T(\p(\chi)-\uniform_{\ab}(\chi)) \\
 & = \frac{1}{4} (H\p  - H\uniform_{\ab})^T(H\p  - H\uniform_{\ab}) \\
 & = \frac{1}{4}(H(\p-\uniform))^T(H(\p-\uniform)) \\
 & = \frac{1}{4}(\p-\uniform)^TH^TH(\p-\uniform) \\
 & = \frac{\ab}{4} \lVert \p-\uniform \rVert_{2}^{2} & (H^TH=\ab I)
\end{align*}
Combining this with an application of Cauchy-Schwartz to relate the total variation distance and the $\ell_2$ norm completes the proof.

\end{proofof}

\subsection{Lower bounds}
    \label{app:lowerbounds}
In this appendix, we provide a proof of our communication lower bound for symmetric (uniformity) testing protocols, restated below:
\begin{proposition}[\cref{thm:lower:bound:communication}, restated]
    \label{thm:lower:bound:communication:restated}
Any \emph{symmetric, private-coin} algorithm for uniformity testing (over domain of size $\ab$) with distance parameter $\dst\leq 1/\ab$ and $\ms=1$ sample per user requires at least $\log_2\ab$ bits of communication per user. (This holds regardless of whether the algorithm is locally private or not.)
\end{proposition}
\begin{proof}
    The argument is nearly identical as that of~\cite[Theorem~2]{AcharyaS:19}. By contradiction, fix any $\numbits$-bit symmetric, private-coin algorithm $A$ for uniformity testing with $1\leq \numbits < \log_2 \ab$ bits of communication per used: and let $W\in \R^{\ab\times 2^{\numbits}}$ be the equivalent representation (of the algorithm used by each user) as a transition probability matrix, where
    \[
        W(x,y) = \probaCond{ A(X) = y }{ X = x}\,.
    \]
    On input distribution $\p$, the distribution of the output at each user is then the probability distribution $W^\top \p$, as, for every $y$,
    \[
    (W^\top \p)(y) = \sum_{x\in [\ab]} W(x,y) \p(x) = \shortexpect_{X\sim \p}[ W(X,y) ]\,.
    \]
    But since $2^{\numbits} < \ab$, the $2^\numbits$-by-$\ab$ matrix $W^\top$ is under-determined, and there exists a non-zero vector $\mathbf{e}$ such that $W^\top\mathbf{e} = \textbf{0}$; without loss of generality, up to rescaling we have $\normone{\mathbf{e}}=2/\ab$. One can further check that $\sum_{x\in [\ab]} \mathbf{e}_x = 0$, as, for every $y\in[2^\numbits]$,
    \[
        0 = (W^\top \mathbf{e})(y) = \sum_{x\in [\ab]} W(x,y) \mathbf{e}_x
    \]
    so that, summing over $y$,
    \[
        0 = \sum_{y\in[2^\numbits]}(W^\top \mathbf{e})(y) = \sum_{x\in [\ab]} \underbrace{\sum_{y\in[2^\numbits]} W(x,y)}_{=1} \mathbf{e}_x = \sum_{x\in [\ab]} \mathbf{e}_x\,.
    \]
    In particular, both the positive and negative parts of $\mathbf{e}$ have $\lp[1]$ norm $1/\ab$, and so $\lp[\infty]$ norm at most $1/\ab$. It then follows that the probability distribution defined by
    \[
    \p \eqdef \uniformOn{\ab} + \mathbf{e}
    \]
    is indeed a valid probability distribution (it sums to one, and is non-negative as $\norminf{\mathbf{e}} \leq 1/\ab$), and is at total variation distance exactly $1/\ab$ from $\uniformOn{\ab}$. As such, it must be distinguished (with high probability) from the uniform distribution $\uniformOn{\ab}$ by any uniformity testing algorithm with distance parameter $\dst \leq 1/\ab$; but $A$ cannot do so, as
    \[
        W^\top \p = W^\top \uniformOn{\ab} + \underbrace{W^\top \mathbf{e}}_{=\mathbf{0}} = W^\top \uniformOn{\ab}
    \]
    that is, the distribution of messages from the users is exactly the same under $\p$ and $\uniformOn{\ab}$.
\end{proof}

\end{document}